\documentclass[format=acmsmall, screen=true]{acmart}
\usepackage[utf8]{inputenc}
\usepackage{acm-ec-25}
\usepackage{booktabs} 
\usepackage[ruled]{algorithm2e} 

\SetAlFnt{\small}
\SetAlCapFnt{\small}
\SetAlCapNameFnt{\small}
\SetAlCapHSkip{0pt}
\IncMargin{-\parindent}
\usepackage{amsmath}
\usepackage{amsthm}

\setcitestyle{authoryear}
\usepackage{subcaption}

\newtheorem{theorem}{Theorem}
\newtheorem{assumption}{Assumption}

\newtheorem{proposition}{Proposition}
\newtheorem{lemma}{Lemma}

\newtheorem{definition}{Definition}

\makeatother

\newcommand{\ut }{\underline{\tau} }
\newcommand{\eps}{\varepsilon}

\title{The Algorithmic Advantage:\\ How Reinforcement Learning Generates Rich Communication}

\author{Emilio Calvano}
\affiliation{%
  \institution{Luiss University, CEPR and Toulouse School of Economics}
  \country{Italy}
}

\author{Clemens Possnig}
\affiliation{%
  \institution{University of Waterloo}
  \country{Canada}
}

\author{Juha Tolvanen}
\affiliation{%
  \institution{University of Rome Tor Vergata}
  \country{Italy}
}

\begin{abstract}

\vspace*{1em}
We analyze strategic communication when advice is generated by a reinforcement-learning algorithm rather than by a fully rational sender. Building on the cheap-talk framework of Crawford and Sobel (1982), an advisor adapts its messages based on payoff feedback, while a decision maker best-responds. We provide a theoretical analysis of the long-run communication outcomes induced by such reward-driven adaptation. With aligned preferences, we establish that learning robustly leads to informative communication even from uninformative initial policies. With misaligned preferences, no stable outcome exists; instead, learning generates cycles that sustain highly informative communication and payoffs exceeding those of any static equilibrium.
\begin{center}
\vspace*{1em}
\vspace{1em}
\end{center}

\end{abstract}

\begin{document}



\maketitle
\pagestyle{plain}


\setlength{\parskip}{\baselineskip}

\section{Introduction}

Suppose an advisor observes a stream of data containing information that is payoff relevant for a decision maker. The advisor employs an algorithm that periodically distills this data stream into a message, transmits it to the decision maker, and observes her subsequent choice. Through this feedback, the algorithm learns to adjust its messaging strategy so as to influence the decision maker’s choices in a direction aligned with the advisor's objectives.

\noindent \textit{Does the algorithm learn to communicate meaningfully? How much information is ultimately transmitted? And does common knowledge of the strategic intent hinder the learning process?}

These questions are of substantial economic importance as algorithmic advice increasingly shapes decisions in diverse domains. For example, large digital platforms extract information from big data and use it to inform and/or influence user behavior/engagement all while users may ignore advice due to biases or incentive misalignment. Recent machine and reinforcement learning methods tackle these challenges jointly, directly optimizing communication policies that induce better decisions without requiring explicit interpretation of the informational content of the underlying data.

To make this concrete, consider a rental platform like Airbnb that provides hosts with daily price suggestions. The platform’s problem is to periodically turn a wealth of time-varying demand data into actionable, unit-specific price advice that hosts will actually find in their best interest to follow. The canonical economics approach is to estimate demand and then compute and effectively communicate the optimal price. However, doing this reliably at scale is complex and model-intensive. Reinforcement learning (RL) offers a more direct, model-free route by learning a mapping from data to price suggestions without explicitly solving the underlying pricing problem. 

To understand how reinforcement learning shapes information transmission, we analyze this problem by embedding an RL algorithm within the well-known cheap talk framework of Crawford and Sobel (1982). Specifically, we assume that the sender acts like a so-called tabular Q-learning algorithm, while the receiver behaves like a sequence of short-lived idealized economic agents always best-responding to the sender's current strategy (or, in one of our extensions, learning faster than the sender). This allows us to talk to applications where human decision makers receive advice from an algorithmic recommendation system learning to best respond to the advice over time. We study both an environment where the receiver's preferences are perfectly aligned with the sender, as well as the traditional environment where the sender prefers higher actions than the receiver.  

The key contribution of this paper is to characterize the simple force that drives the evolution of communication under learning which we refer to as \textit{reward-driven adaptation}. This mechanism reinforces policy changes that generate higher payoffs for the sender and attenuates those that do not, thereby acting as a powerful selection force over communication policies. We show that reward-driven adaptation disciplines the learning process leading to sharp theoretical guarantees and robustly informative communication while never reaching the welfare maximizing level.

In the baseline, no-bias environment we provide analytical results showing that any stable long-run outcome of the learning process must transmit a substantial amount of payoff-relevant information about the state: welfare is bounded away from the babbling benchmark by a tight lower bound. Importantly, this guarantee does not hinge on knife-edge mixing, or favorable initial conditions, and it can arise even when the sender is initialized at a completely uninformative policy. Specifically, using a natural measure of payoff-relevant informativeness, we show that any attracting policy attains at least a tight lower-bound welfare level, which in our baseline calibration equals at least 98 percent of the potential surplus that could be generated under full revelation. 

When the sender prefers higher actions than the receiver, we show that the underlying game does not have any stable outcomes. In particular, the partitional equilibria of Crawford and Sobel (1982) are \textit{not} stable. Hence, any learnt behavior must be changing over time. We then show using simulations that the typical pattern learnt in this environment is one where the algorithm divides the state space using a threshold. Above the threshold, the sender pools all of the states using the same message for all of those states. Below the threshold, however, the sender's policy is constantly changing but very close to full revelation. For all of the currently fully revealing states, the state below will have an incentive to use the message associated with the state above it (due to the positive bias). In other words, the algorithm is constantly trying to inflate its private information below the threshold leading to stable cycling. Nevertheless, this constantly changing language is highly informative. Both the sender's and the receiver's welfare in our simulations is always \textit{above} the best theoretical equilibrium from Crawford and Sobel (1982). 

Another interesting insight is that exploration, while necessary for learning, also \textit{hinders} communication and can prevent the algorithm from ever learning to fully reveal, even in the no-bias case where full revelation maximizes payoffs. Because the sender experiments with non-negligible probability, the receiver understands that even in the best-possible scenario where the sender always truthfully reveals the state, any given message could have been generated from other states due to experimentation, so her choice will be more `cautious,' as the posterior is pulled toward the prior. When the learning algorithm infers that the receiver does not fully follow its advice, it responds by adjusting messages to correct for these more cautious actions, which makes exaggeration profitable and leads the dynamics toward pooling at the extremes of the state space. In the no-bias environment, we argue that , for sufficiently high exploration rates, fully revealing policies are not stable under learning, and the simulations confirm that the long-run outcome is a highly but not fully informative pooling structure.

Needless to say, our focus on tabular Q-learning raises external validity concerns. We nevertheless adopt it for three reasons. First, tabular Q-learning is a natural benchmark for our setting: it is widely used as a baseline method for learning optimal policies from payoff feedback.\footnote{Other natural candidates for this problem are contextual bandit algorithms; it turns out that Q-learning in this setting is equivalent to contextual versions of the `Exp3' bandit algorithm \cite{auer2002nonstochastic}.}
Second, its transparency makes it suitable for isolating the mechanisms through which reward-driven updating shapes strategic communication. Finally, it is a foundational building block for more sophisticated reinforcement-learning architectures, so studying it allows us to isolate forces that plausibly carry over to broader classes of learning-based systems. 


\section{Related Literature}
Algorithmic advice from AI or machine learning tools has become increasingly important to human decision-makers in a diverse and broadening array of domains from the judiciary system \citep{ludwig2021fragile} and employee recruitment processes \citep{hoffman2018discretion} to recommending pricing for accommodation services \citep[e.g.,][]{huang2022pricing, garcia2022demand}, airline tickets, and recommendations on platforms such as Amazon, Spotify, and Netflix. These services must not only gather relevant information but also present it to users in a credible and profitable manner so as to boost engagement. 

A number of authors have recently argued that the algorithm's designer's and the end user's incentives are often misaligned \citep[e.g.][]{cowgill2020algorithmic, huang2022pricing, garcia2023strategic}. This makes algorithmic advice a problem of strategic communication rather than mere prediction. Zooming in on this communication aspect, a paper developed concurrently with this one, \citet{condorelli2023cheap}, use simulations to study a cheap-talk environment where both the sender and the receiver are reinforcement learners and document substantial information transmission.  We provide the first theoretical analysis of of reinforcement-learning–based cheap talk. Our analysis considers a version of the problem in which the receiver learns faster than the sender, and delivers theoretical guarantees that isolate the reward-driven mechanisms through which learning disciplines communication and ensures the emergence of an informative language.


Somewhat similarly, computer scientists have considered whether algorithms can learn to communicate with each other \citep[see especially][]{lazaridou2017multiagent, NIPS2017_70222949, noukhovitch2021emergent}. Relative to our work, the papers we know from this area are simulation-based, use machine learning for both the sender and the receiver, and do not consider the exact standard cheap talk framework of \citet{crawford1982strategic}. They do highlight that computer scientists consider solving communication problems using machine learning and reinforcement learning in particular. We contribute to this literature by providing a strict theoretical lower bound on welfare for the standard cheap talk problem with aligned interests even when the sender uses a simple soft-max Q-learning, and we highlight strategic frictions that arise both with aligned and misaligned interests when the receiver is sophisticated but myopic.\footnote{We believe this to be a natural benchmark for most recommendation systems where the receiver is a human. Exploration in the algorithm will often be observable to humans from the realized outcomes and hence they will adapt their expectations accordingly. However, the same algorithm is typically used for a large population of receivers and hence there is little hope for a single receiver to try to ``teach'' the algorithm, justifying our assumption on myopia.}

Our paper is also related to research in evolutionary game theory that considers the emergence of language. \citet{WARNERYD1993532} shows that in a cheap talk game where the sender is picked at random, where only pure strategies are allowed and where the two players have identical preferences, the equilibria that satisfy a weak version of evolutionary stability are fully informative equilibria. In contrast, we argue that when learning happens through experimentation, that experimentation generates additional frictions and the final languages are typically not fully informative. Furthermore, we show that evolution through experimentation in the sense of reinforcement learning will never lead to stable languages when the sender's and the receiver's preferences are misaligned. There is also a body of work that considers learning languages through a lens of replicator dynamics that resemble our theoretical tools \citep[for a review of this literature, see][]{nowak2002computational}. Reminiscent to our results on Q-learners, \citet{PAWLOWITSCH2008203} shows that the evolutionary replicator dynamics can get stuck at suboptimal languages even when the sender's and receiver's preferences are aligned as long as mixed strategies are allowed. Relative to those results, we provide a lower bound on welfare when evolution (and hence sender's mixing) happens through Q-learning, and show that no equilibria can be learnt when the sender's and receiver's preferences are misaligned.

Last, our paper contributes to the rapidly increasing literature on how algorithmic systems interact with themselves and/or with humans through markets or other social institutions. A large fraction of this literature studies why and how reinforcement algorithms are able to obtain payoffs that are above any stage-game payoff \citep[see, for example,][]{calvano2020artificial,asker2022artificial, banchio2022artificial, Johnson2023,possnig2023reinforcement, brown2023}. For example, in pricing games supra equilibrium payoffs are tantamount to collusion between firms. Our results on the case of preference disagreement add to this literature by showing that payoffs that are higher than the best equilibrium payoff can be obtained also in the context of cheap talk. Our result differs from the work on collusion in markets in the sense that our algorithm does not condition on a history of actions and the observed outcomes are not equilibria of a dynamic game but stable best-response cycles.

At the same time, our analysis highlights an important limitation of reinforcement-learning–based systems. When sender and receiver interests are aligned, learning does not generally lead to full welfare maximization: exploration and reward-driven adaptation can prevent convergence to the fully revealing equilibrium, even when it is efficient. In this sense, the same features that make reinforcement-learning algorithms attractive in practice, scalability, model-free implementation, and robustness, also come with strings attached. This negative result connects our findings to recent work showing that algorithmic systems designed to assist users in their decisions can systematically distort behavior and lead to suboptimal outcomes, including in recommender systems \citep{calvano2025artificial}.

\section{The Model}
We consider the following discrete version of \citet{crawford1982strategic}'s quadratic payoff setting for our interaction between an algorithm and a rational agent.\footnote{We consider a discretization since our algorithmic sender requires a game with finite action and state spaces.}  There is a sender $S$, who observes a state variable $X$, uniformly distributed on $\mathbf{S}_K = \{x^{(1)},\dots,x^{(k)},\ldots, x^{(K)}\} \subset [0,1]$, for some $K\in \mathbb{N}$. We fix $x^{(1)} = 0, x^{(K)} =1$, so that $x^{(k)} - x^{(k-1)} = \frac{1}{K-1}$ for all $1<k\le K$.

Given a realized state $x$, the sender's payoff is $U^S(y,x,b) = -(y - x - b)^2$, where $y$ represents the receiver $R$'s action, and $b>0$ represents the preference disagreement between the receiver and the sender that we will often call the sender's bias.\footnote{The model is identical to one with a negative $b$ or one with an additive bias term for both the sender and the receiver in which case only their difference matters.} The receiver's payoff given state realization $x$ and action $y$ is $U^R(y,x)= -(y - x)^2$. After observing the state $x$, $S$ sends a message $m \in \mathbf{M}_K=\mathbf{S}_K$ to the receiver, who considers the message, updates their belief over the state $X$, and then plays an action $y\in \mathbb{R}$. The sender has a messaging strategy $\mu(x,m)$ which specifies the likelihood of sending message $m \in M$ given state $x$. The receiver, in turn, has an action strategy $y(m,\mu)$ specifying an action for every possible message sent by $S$, given $S$'s policy $\mu$. We refer to the space of mixed sender strategies $\mu$ as $\Theta$, and to the game as $\Gamma$.

\subsection{Algorithmic Cheap Talk}\label{subsec: Algo Intro}
This is where our study departs from the classical cheap talk setting among rational agents. We consider an interaction in which the sender is not fully rational, but a reinforcement learning algorithm. This algorithm updates its behavior over time, as payoffs from an interaction with a rational receiver are accrued. 

We consider a tabular Q-learner as the sender $S$. Q-learning is a much studied reinforcement learning method, well regarded due to its simplicity, analytical tractability and minimal information requirements. For a detailed introduction, see \citet{sutton2018reinforcement}. Q-learning was introduced to find optimal policies for single-agent Markov decision problems (MDPs), when little to no information about payoff functions and state transitions is known to the algorithm's designer. This learning rule has become the focus of attention also of researchers interested in multi-agent learning, again due to its minimal modeling requirements, ease of setup, and also since many more involved reinforcement learning schemes retain, at the very least, some of the conceptual intuitions introduced by Q-learning. 

 A Q-learner estimates what is known as a Q-value function. This function $Q: \mathbf{S}_K \times \mathbf{M}_K \to \mathbb{R}$ is meant to allow the learner to find an optimal policy mapping from $\mathbf{S}_K$ to $\mathbf{M}_K$. Essentially, this function intends to estimate the expected payoff for any given state and action, allowing to find optimal actions upon realization of a state. The estimator is formed as a simple weighted-averaging rule, and only updated upon realizing a given state-action pair.

We refer to realizations at a period $t$ using subscripts. Every period $t$, a state $x_t$ is realized and the learner selects a message $m_t$ at random (more below). The receiver then takes their posterior as action $y_t = y(m_t,\mu_t)$, and payoffs are realized. Importantly, we assume $u_t = U^S(y_t,x_t,b)+\eta$, where $\eta $ is a zero-mean noise term with bounded variance $0<\sigma_{\eta}^2<\infty$. This noise term both serves as helpful tool in our theoretical approach, and to robustify our results to the realistic scenario in which a learner's payoffs are only noisily observed.  Upon a realization of state $x_t$, message $m_t$, and payoff $u_t$, the Q-function is updated the following way:
\begin{align}\label{eq: Qlearn}
	    Q_{t+1}(x,m) 
        &= 
	    (1-\alpha_t)Q_t(x,m) + \alpha_t \Big[u_t + \beta \max_{m' \in \mathbf{M}_K}Q_t(x_{t+1},m')  \Big],
	\end{align}
    for $(x,m)=(x_t, m_t)$, and $Q_{t+1}(x,m)=Q_t(x,m)$ otherwise. 
	Here, $\alpha_t> 0$ is a non-increasing sequence,  known as the `learning rate', and $\beta\in [0,1]$ is a discount factor. Our framework carries no time-dependence among state realizations, as states are sampled every period i.i.d. uniformly from $\mathbf{S}_K$. In our baseline setting, we therefore set $\beta=0$.\footnote{We confirm in our simulation study that $\beta>0$ does not qualitatively alter our results.}
 
 Notice that $Q$-learning does not specify a policy. Common convergence results on $Q_t$ give requirements on how often actions are selected over time, but the updating rule is agnostic about how actions $a_t$ are sampled in every period. Mainly, we consider one of the most common sampling rules: the softmax, also known as `Boltzmann'-Q- learning. Under this choice rule, actions are selected proportional to their payoff estimates:
 \begin{align*}
    \mathbb{P}(m_t = m|x,Q) &= \bar \mu(x,m,Q) = (1-\varepsilon_t)\mu(x,m,Q) + \varepsilon_t \frac{1}{K},\\
    \mu(x,m,Q) &= \frac{\exp(\tau_t^{-1}Q(x,m))}{\sum_{m'}\exp(\tau_t^{-1}Q(x,m'))}.
\end{align*}
 $0<\tau_t\le 1$ is a non-increasing sequence that controls the weight put on the maximizing action, while $\varepsilon_t\in (0,1)$ is a non-increasing sequence of weights on a uniform distribution over messages. One can think of $\varepsilon_T$ as ensuring a minimal positive mass on all messages, no matter the value of $\tau$. In what follows, we will be commonly assuming that $\lim_{t\rightarrow \infty}\tau_t = \underline{\tau}>0$, $\lim_{t\rightarrow \infty}\varepsilon_t = \varepsilon>0$.\footnote{We also ran simulation experiments with the so called $\varepsilon$-greedy learner where $\tau_t=0$ for all $t$. The results remain qualitatively very similar and are available from the authors on request.} Throughout, for our theoretical results, we impose the following assumption on hyperparameters of our algorithm:
 \begin{assumption}[Q]
     \begin{enumerate}
     \item[] 
         \item Step sizes satisfy
         \begin{align*}
             \sum_{t\ge 0}\alpha_t = \infty && \sum_{t\ge 0}\alpha_t^2 < \infty.
         \end{align*}
         \item As $\underline{\tau}>0, \eps>0$ vanish, $\frac{\ut}{\eps^2}$ vanishes also.
     \end{enumerate}
 \end{assumption}
 Point (1) is a standard assumption in the theory of reinforcement learning algorithms, ensuring on the one hand that stochasticity is averaged out (fast enough convergence of $\alpha_t$), while allowing for the algorithm to traverse the whole space. Point (2)  ensures that our stability analysis remains valid; our simulations verify that the results extend to more general relationships between $\ut, \eps$.

As it will become useful later on, we distinguish between $\bar \mu()$, S's full messaging policy, and $\mu()$, S's messaging policy excluding exploration parameter $\varepsilon_t$.  To save notation, we sometimes write $\mu_t(x,m)=\mu(x,m,Q_t)$. The receiver $R$ is what we consider a `sophisticated, short-lived' agent: at every period $t$, $R$ knows the sender's policy $\mu_t$ but does not try to manipulate the learning process of the algorithm. This setting can be seen as an extreme version of a receiver who knows that their recommendations are being generated algorithmically and learns the algorithm's strategy faster than the time between the updates to the algorithm's policy.\footnote{We show in Appendix \ref{app: leanringR} that our results are robust to a setting where the receiver learns, according to a family of rules that `update faster' than the Q-learner.} We believe it is a natural starting point as it allows us to highlight how strategic communication is affected when only the sender is algorithmically learning and a rational receiver anticipates this. However, like in the case of AirBnB, we assume that the receiver is part of a much larger population of receivers or otherwise myopic and hence believes their actions have no perceivable impact on the algorithm's learning.\footnote{We leave the case where the receiver tries to teach the algorithm for future research. We conjecture that a patient, rational receiver can eventually implement their preferred policy at least when $\beta>0$, as then the receiver can easily penalize the algorithm for using undesired messaging strategies.} Upon receiving a message $m_t$, the receiver forms a belief about the true state of the world. We write $y_t(m) = y(m, \mu_t)$ as the receiver's strategy if there is no potential for confusion. The quadratic payoff function implies that, given $\mu_t$ and the sender's message $m$, the receiver's best response is to set $y_t(m)$ to be the expectation of $x$ conditional on $(\mu_t, m)$: $y_t(m) = \mathbb{E}\left[X\mid m,\bar \mu_t \right]$. As the receiver's best response tracks $Q$-values (or a policy $\bar \mu()$ given $Q$-values) by assumption, one can write the sender's payoffs as $u(x,m,Q_t)=-\left(x+b-y_t(m)\right)^2$.  We sometimes write $\mu(Q)$ as the softmax policy pinned down by $Q$-values $Q$. 

The updating rule \eqref{eq: Qlearn} together with the fact that the receiver does not know whether at any point the algorithm is exploring or not, imply that this learner can only implement cheap talk. The sender estimates deviation payoffs through exploration - which can happen any period, with positive probability. The receiver cannot react to such a `deviation', as it is unobserved by them. For Bayesian persuasion to be implemented, the sender would have to receive feedback from the receiver through a best response to the sender's deviation, which is impossible in this setting. Furthermore, unlike in Bayesian persuasion, the sender here does not take into account how a change in its policy in one state changes its payoffs in other states.

\section{Preliminary results}

In each period, after observing the current state, the sender draws a message according to a softmax policy over current $Q$-values and then updates the $Q$-value of the realized state-message pair proportionally to the realized payoff. As the step size and the temperature decay, the discrete-time $Q$-learning dynamics can be approximated by an ordinary differential equation in the space of $Q$-values, and the associated limiting policies are stationary points of the induced policy dynamics (Proposition \ref{prop: stochapproxQ} in Appendix \ref{app: nashconv}). We are interested in such limiting policies.


 \begin{definition}
     \begin{enumerate}
     \item[]
     \item Let $\mathbf{NE}_{CS}$ be the set of sender's strategies that are part of a perfect Bayesian equilibrium of $\Gamma$.
     \item For any $\zeta>0$, and sender policy $\mu^*$ define $\mathbf{B}_{\zeta}(\mu^*) = \{\mu \in \Theta: \|\mu - \mu^*\|_{\infty}<\zeta\}$ as the $\zeta$-ball around $\mu^*$.
        \item $\mathbf{L}(\mu_0)$ be the set of limit points of convergent subsequences of $\mu_t=\mu(Q_t)$, given initial value $\mu_0$:
\begin{align*}
    \mathbf{L}(\mu_0) = \bigcap_{t\ge 0}cl\left( \{\mu_{\ell} \mid \ell \ge t\}\right),
\end{align*}
where for any set $A$, $cl(A)$ is the closure.
    \end{enumerate}
\end{definition}
As the learning algorithm we study will always select suboptimal messages as long as limiting temperature $\ut$ and exploration probability $\eps$ remain positive, we must allow for the possibility of learning a policy that is close but not exactly an equilibrium in $\mathbf{NE}_{CS}$, which we measure using $\mathbf{B}_{\zeta}$. The first result will show that the approximation can be made arbitrarily fine. 

$\mathbf{L}(\mu_0)$ is the set of convergent subsequences of a single run of the learning algorithm. This definition allows for non-convergence, which will become important in the bias case. As $\mu_t$, is stochastic, also, $\mathbf{L}(\mu_0)$, the set of limit points of convergent subsequences of $\mu_t$ will be randomly determined.\footnote{The set is always non-empty due to the Bolzano-Weierstrass theorem.} Hence, all of our results will be probabilistic.

\begin{theorem}\label{thm: Nashconv}
 For all $\zeta>0$ there exists $\ut ,\varepsilon>0$ small enough such that for all $\mu_0$,         $\mathbf{L}(\mu_0) = \{\mu^*\}$  implies there exists  $\mu' \in \mathbf{NE}_{CS}$ such that $\mu^* \in \mathbf{B}_{\zeta}(\mu ')$   almost surely.
\end{theorem}
In other words, convergence of the algorithm implies that a policy arbitrarly close to a perfect Bayesian equilibrium of the discretized cheap talk game is being played. 

Notice that our learning setup cannot sustain equilibria that rely on finely tuned mixing over payoff-equivalent messages. In many Crawford–Sobel equilibria, the sender is indifferent across multiple messages either within a state or across states, and equilibrium play is supported only by selecting specific probabilities over those messages. A softmax Q-learner cannot reproduce such behavior. By construction, any two actions that yield the same continuation payoff are assigned equal probability in the limit, so the learner cannot maintain the asymmetric mixing required to support these knife-edge equilibria. Consequently, limiting policies can only correspond to equilibria that do not rely on asymmetric mixing across payoff-equivalent messages.

\section{No Bias}
Consider the case when $b=0$ and hence the incentives of the sender and the receiver are fully aligned. In this benchmark, the static Crawford-Sobel game admits a fully revealing equilibrium, a babbling equilibrium, as well as a large number of other equilibria that are somewhere in between these two extreme cases.\footnote{Babbling refers to an outcome in which the sender’s messages carry no information about the state, so the receiver ignores them and always chooses the same action.} Our main result is that algorithmic communication is guaranteed to be highly informative but it almost never achieves \textit{fully} informative communication.  

The key conceptual feature of this dynamic is that it is reward-driven: whenever a small change in the policy increases the sender's expected payoff, the induced change in $Q$-values is such that the policy drifts in that direction. This reward-driven property naturally leads the learner to discard highly uninformative outcomes. As an extreme example, consider babbling, the least informative equilibrium. Suppose the algorithm is initialized at a babbling policy and hence the receiver's action is also independent of the state. Furthermore, to stack the deck against information transmission, assume that the babbling policy is induced by $Q$-values all equal to the payoff that the receiver gets from the mean action $\tfrac{1}{2}$ in each corresponding state. In other words, assume that the algorithm starts at correct payoff estimates and strategy corresponding to a babbling equilibrium of the underlying Crawford-Sobel game where all actions in all states are played with equal probability. Now consider what happens if one of those $Q$-values is slightly perturbed to be higher in state $x$. Whenever the receiver sees the message corresponding to the perturbed Q-value, they will now associate an action to it that is slightly closer to $x$, yielding a higher payoff. This creates a feedback loop: the sender reinforces the perturbed message, which further increases its $Q$-value and raises the probability of playing it in that state. Simultaneously, also in states far away from $x$ the algorithm learns that the action the receiver takes after seeing the perturbed message is worse than what is available by using the other babbling messages reinforcing other messages. Hence, small local deviations away from babbling are strictly payoff-improving, and reward-driven adaptation pushes the policy away from the babbling equilibrium. In other words, the algorithm naturally unravels babbling.

The same logic can be used more generally. Any policy that can be locally perturbed in a way that raises the sender's expected payoff cannot be a stable limit point of the $Q$-learning dynamics. 

One might be tempted to conclude that the only thing that the algorithm can learn is then full revelation as it maximizes the sender's payoff. We will show that this intuition is not true as there exist equilibria of the game where local perturbations, even if observed by the receiver, will not improve the sender's (or the receiver's) payoff, leading to the algorithm unlearning them. Yet, this reward-driven adaptation is going to rule out a large number of low information equilibria. In what follows, we show that it implies three necessary structural restrictions on candidate limiting policies which we then use to construct our information lower bound. 

The first property is that limiting policies must be such that a single message cannot be used on ``disconnected'' subsets of the state space.

\begin{definition}[Connected policy]
A policy $\mu$ is connected if for any message $m$ and any $x<x'$ such that $\mu(x,m)>0$ and $\mu(x',m)>0$, we have $\mu(x'',m)>0$ for all $x''\in[x,x']$.
\end{definition}

 Connectedness may seem natural in this context as Crawford and Sobel (1982) already showed that all equilibria of the cheap talk game with a \textit{positive bias} are connected. Notice, however, that this is not true in their setting when the bias is zero. As an example, consider the sender sending a single message for the most extreme right and left states and pooling the rest of the states under another message. This sender's strategy would yield a disconnected perfect Bayesian equilibrium of the cheap talk game with zero bias.

 Intuition for why limiting policies for the algorithm must be connected can be built as follows. Suppose that there is a limiting policy such that for some message $m$ the algorithm uses $m$ at states $x$ and $x'$ but assigns zero probability to using $m$ at any state strictly between $x$ and $x'$. If playing $m$ at $x$ and $x'$ is optimal against the receiver's response to $m$, then either (i.) playing $m$ at some intermediate $x''\in(x,x')$ leads to a weakly higher payoff or (ii.) at either $x$ or $x'$ there is another message $m'$ that leads to a strictly higher payoff than $m$. To see this, pick first a state $x''\in (x,x')$ and consider what happens if the posterior of the receiver satisfies $y(m)< x''$. Then playing $m$ leads to a weakly larger payoff for the sender in $x''$ than in $x'$. Conversely if $y(m) > x''$ then playing $m$ leads to a weakly higher payoff in $x''$ than in $x$. So playing $m$ in $x''$ is not optimal for the sender only if there is a message $m''$ such that the expected payoff from that message is larger implying in turn that the receiver's posterior after that message must be closer to $x''$ than after $m$. That implies that message $m''$ yields a strictly larger payoff either in $x$ or in $x'$ contradicting the working hypothesis that $m$ is optimal in those two states. If instead $y(m)=x''$, the algorithm will learn to play $m$ at that state with a probability that is at least as high as what it uses for other messages in that state contradicting the assumption that $m$ is used only in states $x$ and $x'$. 
 
 In summary, if a non-connected policy is not an equilibrium, the algorithm will eventually learn the deviation that benefits it and will reinforce that deviation. That is, any ``hole'' between two states that both send the same message is eliminated by reward-driven adaptation. If the policy is an equilibrium, it requires the sender to be indifferent across messages used in states in the disconnected part of the policy and its complement. The algorithm will eventually learn to play those messages at equal probabilities in all of the states where the equilibrium requires indifference, again filling the ``hole''.


The second property that limiting policies must satisfy is that the middle state is not pooled with any other state. 
\begin{definition}[Middle state fully revealing, MSFR]
A policy $\mu$ has property MSFR, if for message $m$ and $x\neq 0.5$,  $\mu(0.5,m)>0$ implies $\mu(x,m)=0$.
\end{definition}
That is, reward-driven adaptation makes sure that when there is no news this is revealed to the receiver.

The simplest intuition for why this is necessary can be built by starting from a candidate policy that sends some message with arbitrarily small probability, i.e. features an ``unused'' message, while prescribing an action different from $1/2$ at the mean state $1/2$. Under this policy, playing an unused message conveys no news and therefore leads the receiver to choose the action that is optimal given only the prior. This action yields the correct payoff at the mean state $1/2$, and so any small amount of experimentation on this unused message raises its $Q$-value relative to the messages that implement an action which is different than $1/2$. As a consequence, the learning dynamic increases the probability of sending the unused message at state $1/2$, and the policy drifts toward revealing that the state equals the prior. In this way, reward-driven adaptation naturally eliminates any candidate policy that prescribes a suboptimal action in the middle state. 

A very similar, but more complex logic applies when there is no such unused message. In that case every message is sent with positive probability, but it is still possible to show that small differences in relative message probabilities create the same type of payoff-improving direction. Messages that are played slightly less often induce posteriors that are closer to the prior, and therefore yield higher payoffs at the middle state. These small differences are then amplified by the learning dynamic, which reinforces the more profitable message and shifts probability mass toward it.

MSFR is true in general, but it is especially natural in applications where the sender has the option to `remain silent' / not provide any recommendation (which is of course itself a message) and the receiver is naive in the sense that silence is interpreted as an absence of new information. In such cases, the receiver chooses the action that is optimal given only the prior, and the algorithm quickly learns to `remain silent' whenever the state coincides with that prior. In other words, with naive receivers the MSFR condition arises almost automatically.

The final property requires that when the sender suppresses information by sending the same message across multiple states, this suppression be applied consistently and in a not too dissimilar fashion across adjacent states.

\begin{definition}[Similar adjacent pool sizes, SAPS]
A policy $\mu$ has property SAPS, if for any $x<1$, 
\begin{eqnarray}
    &&\left|\left\{x':\exists m \textrm{ s.t. } \mu(x+1/(K-1),m)>0 \textrm{ and } \mu(x',m)>0\right\}\right|-1\nonumber\\
    &\leq& \left|\left\{x':\exists m \textrm{ s.t. } \mu(x,m)>0 \textrm{ and } \mu(x',m)>0\right\}\right|\nonumber\\
    &\leq& \left|\left\{x':\exists m \textrm{ s.t. } \mu(x+1/(K-1),m)>0 \textrm{ and } \mu(x',m)>0\right\}\right|+1\nonumber
    \end{eqnarray}
\end{definition}

SAPS is an incentive-compatibility requirement that neighboring pools must be roughly similar in size, otherwise types at the boundary of a large pool would prefer to mimic those in the adjacent much smaller pool. Specifically, the total number of states in two adjacent pools can differ by at most one. This property reflects the strategic environment of the discrete Crawford-Sobel communication game: large disparities in pool sizes create incentive issues that make such pooling patterns unstable. The learning dynamic inherits this feature and imposes it in an even stronger form.

 As a concrete example consider a system like AirBnB giving price recommendations to its hosts. Suppose the algorithm currently plays the following connected strategy where message 1 signals an optimal price of 10, while message 2 indicates that the optimal price is above 10. If the optimal price can take many values greater than 10, the sender will prefer, for optimal prices just above 10, to deviate from message 2 to message 1, as this leads the receiver to set a price that is closer to the true optimal price.

The following theorem takes stock. 
\begin{theorem}\label{thm: attrchar}
        For all $\zeta>0$ there is $\ut, \eps$ are small enough, so that for all initial $\mu_0$, and  $\mu^*\in \mathbf{NE}_{CS}$, we have that:\\ $\mathbb{P}\left(\mathbf{L}(\mu_0) \subseteq \mathbf{B}_{\zeta}(\mu^*) \right)>0$ if and only if $\mu^*$ is Connected and satisfies MSFR and SAPS. 
    \end{theorem}
In the proof (relegated to the appendix) we argue that the equilibria which the dynamic learning process can converge to are attractors of a differential equation and that those attractors will be Connected and satisfy MSFR and SAPS when $\ut $ and $\varepsilon$ are small enough. 

To compare policies, we measure the payoff-relevant informativeness of a policy by the expected utility that the sender and the receiver obtain from it plus the variance of the state, $\frac{1}{12}\frac{K+1}{K-1}$, and normalized by this same variance. If the receiver does not have access to any recommendations, their maximal expected utility is equal to the negative of the variance of the state, $-\frac{1}{12}\frac{K+1}{K-1}$. By adding this to the expected utility and normalizing with it, the new measure ranges from 0 to 1. Formally, our normalized measure of welfare-relevant informativeness of a policy $\mu$ is given by
$$U(\mu)=\frac{\frac{1}{12}\frac{K+1}{K-1}-\mathbb{E}[(X-\mathbb{E}[X\mid m, \mu])^2]}{\frac{1}{12}\frac{K+1}{K-1}},$$
where the outer expectation is taken over the joint distribution of $m$ and $X$ induced by the messaging strategy $\mu$. Notice that this measure of informativeness is coarser than, for example the correlation between the state and the messaging strategy or entropy based measures such as the mutual information between the two. The reason for this is that there exist many strategies, and even equilibria of the underlying cheap talk game that can differ substantially in terms of their informativeness when measured by these more classical information measures but have equivalent payoff consequences for both parties. For example, consider a strategy where the sender sends message $0$ for both state $0$ and state $1$ and message $1$ for all other states. The posterior of the receiver after seeing either of those messages is $0.5$ and hence the receiver's action will be the same after both of them and this action coincides with the action the receiver would take were the sender to babble and send all messages with an equal probability after every state. Nevertheless, the messaging strategy above is considerably more informative than the babbling strategy - the information it conveys is just not payoff-relevant for either of the parties. 

We are now ready to state the main result of this section.
 \begin{theorem}\label{thm:infobound}
           There is a positive function $\nu(K)=O\left(\frac{1}{K}\right)$ with $\nu(K) \in [0,1]$  such that for all grid sizes $0<K<\infty$ there exists welfare threshold $\underline{U}_K\in (0,1)$ and  for all $\zeta>0$ there are $\ut>0, \varepsilon>0$ small enough:
            \begin{enumerate}
                \item For all $\mu_0$, if $ \mathbf{L}(\mu_0)=\{\mu^*\}$, then
            \begin{align*}
               U(\mu^*)\ge \underline{U}_K,
            \end{align*}
            with probability $1$. 
            \item The exists $\underline{\mu}\in \mathbf{NE}_{CS}$ with $U(\underline{\mu})=\underline{U}_K$.\newline For a positive measure set of $\mu_0$, $\mathbb{P}\left[\mathbf{L}(\mu_0) \subseteq \mathbf{B}_{\zeta}(\underline{\mu}) \right]>0$.
            \item $\underline{U}_K = 1-\nu(K)$.
            \end{enumerate}
        \end{theorem}

The first part of the theorem shows that there exists a lower bound on the information that the algorithm may learn, giving a robust guarantee for the informativeness of languages learnt by an RL algorithm. The second part shows that the lower bound is learnt with a positive probability, making it tight. The last part shows that when the grid becomes increasingly fine, the lower bound converges to $1$ from below. In other words, the more complicated the messaging problem becomes, the better is the guarantee. 

The proof of the result is constructive.\footnote{See Appendix \ref{app: attrcharproof}.} We use the three properties presented above to construct a worst-case scenario and show that that worst-scenario can always be learnt with positive probability. More specifically, as required by MSFR, we start from a fully revealing strategy in the middle. Connectedness then requires that any pools on either side of the middle state must be connected. Following SAPS, we then construct maximal pools by bunching subsequent states away from the middle into pools so that neighboring pool sizes differ by exactly 1. As $K$ grows, it will commonly not be possible to `fit' a sequence of pools with this property into the available state grid. Nevertheless, we show that the feasible worst-case policy closely tracks such a construction, and importantly, the largest pool that can be achieved in such a policy only grows at a rate of $\sqrt{K}$. This argument then yields a lower bound on the welfare of the cheap talk game that is captured by Theorem \ref{thm:infobound}, which approaches $1$ as $K$ grows. Two examples of worst-case policies that may be learnt according to this Theorem are shown in Figure \ref{fig: lbpols}.

\begin{figure}[t]
\centering
\begin{subfigure}{.5\textwidth}
  \centering
  \includegraphics[width=.9\linewidth]{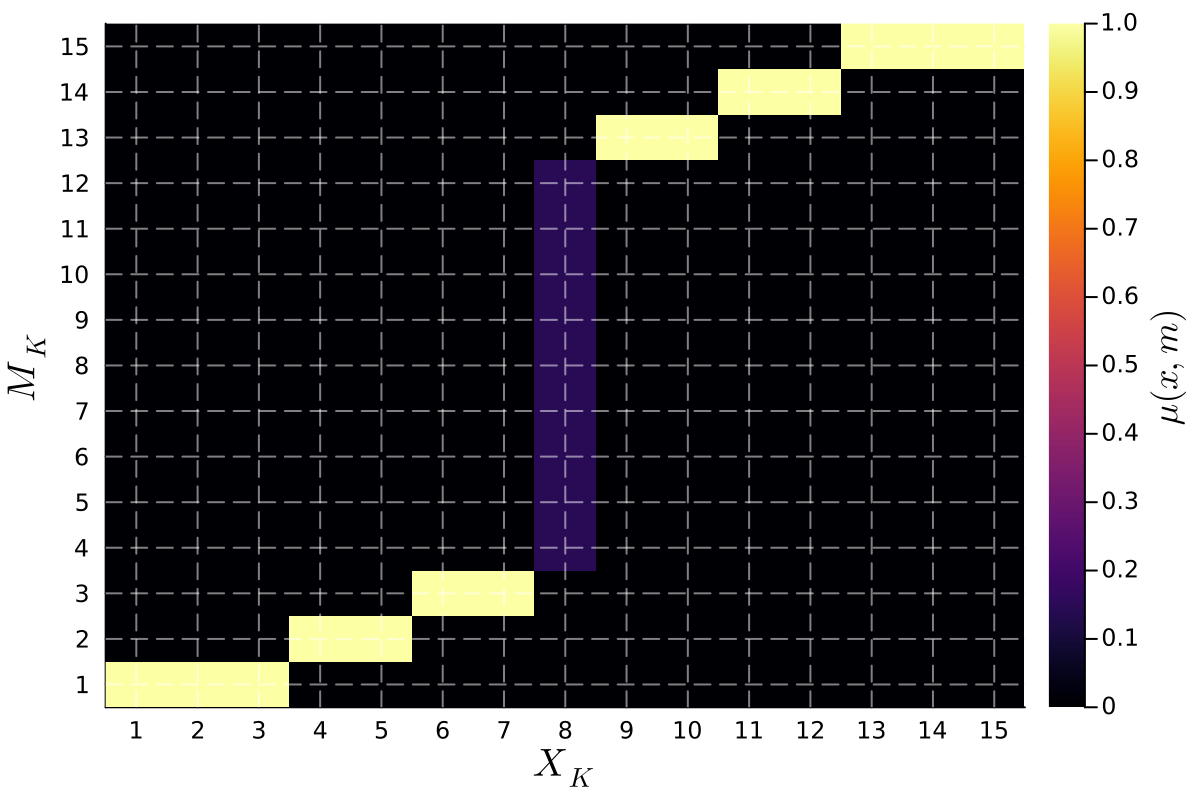}
  \caption{\footnotesize }
  \label{fig: lbpol K15}
\end{subfigure}%
\begin{subfigure}{.5\textwidth}
  \centering
  \includegraphics[width=.9\linewidth]{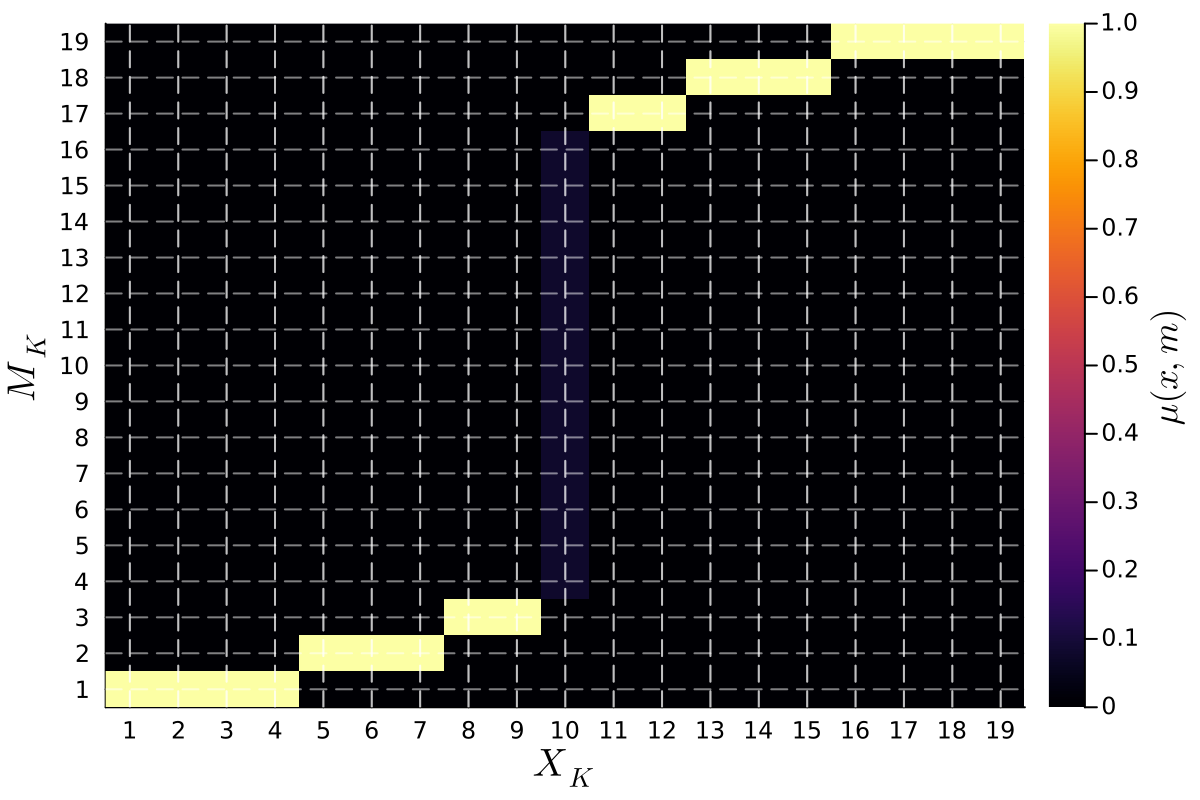}
  \caption{\footnotesize}
  \label{fig: lbpol K19}
\end{subfigure}
\caption{\footnotesize Worst-case policies achieving $\underline{U}_K$ for $K\in \{15,19\}$, represented by a heatmap. The $x$-axis refers to states $x\in X_K$, $y$-axis to messages $m\in M_K$, while color indicates probability mass.  In (a), $K=15$ and the grid cannot fit the sequence of pools that always increase by $1$, so the construction repeats pool sizes. The worst case was computed by brute force according to the constraints given by SAPS, MSFR, and connectedness.  In (b),  where $K=19$, we have an example where the sequence fits exactly, and no computation is necessary. We implement the policy with middle state ($8$ in (a), $10$ in (b)) fully revealing by mixing over messages, as is typically learnt by the algorithm as shown in Figure \ref{fig: nobias_pol}.}
\label{fig: lbpols}
\end{figure}

\textit{How often does the algorithmic learner hit the information lower bound in practice?} 

This is a question hard to tackle in theory but straightforward to answer in practice running simulations experiments. We set $\ut=10^{-4}$, $\varepsilon=10^{-3}$, $K=21$, $\alpha_t=0.05$ for all $t$ and we take $\tau$ to $\ut$ at rate $e^{-\gamma t}$, for $\gamma = 10^{-4}$. We simulate the game $1000$ times until convergence starting it both from a fully babbling and a fully revealing initial strategy. All our simulation runs converge, Theorem \ref{thm: Nashconv} ensuring that every resulting limit point is a Perfect Bayesian Equilibrium of the underlying cheap-talk game.

Figure \ref{fig: nobias_heat} shows the histogram of $U(\mu)$ over runs. 
\begin{figure}[t]
\centering
\begin{subfigure}{.5\textwidth}
  \centering
  \includegraphics[width=.9\linewidth]{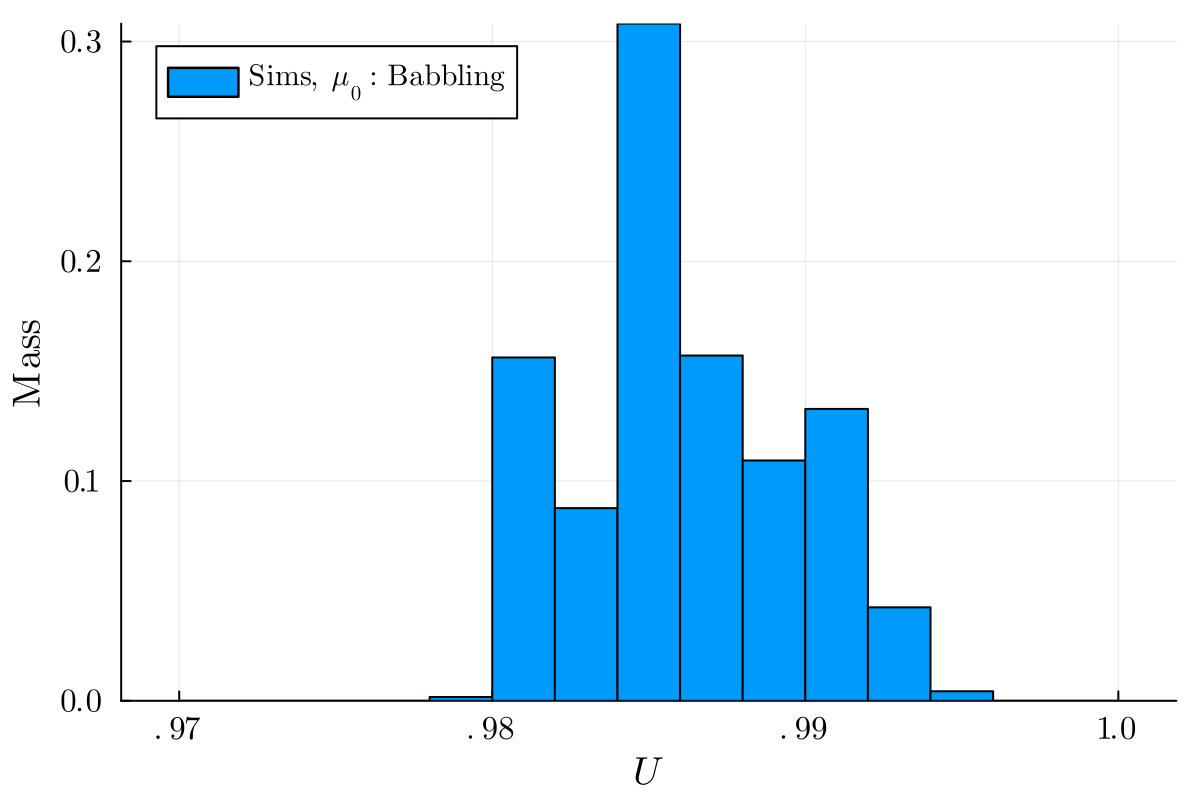}
  \caption{\footnotesize }
  \label{fig: nobias_heatbab}
\end{subfigure}%
\begin{subfigure}{.5\textwidth}
  \centering
  \includegraphics[width=.9\linewidth]{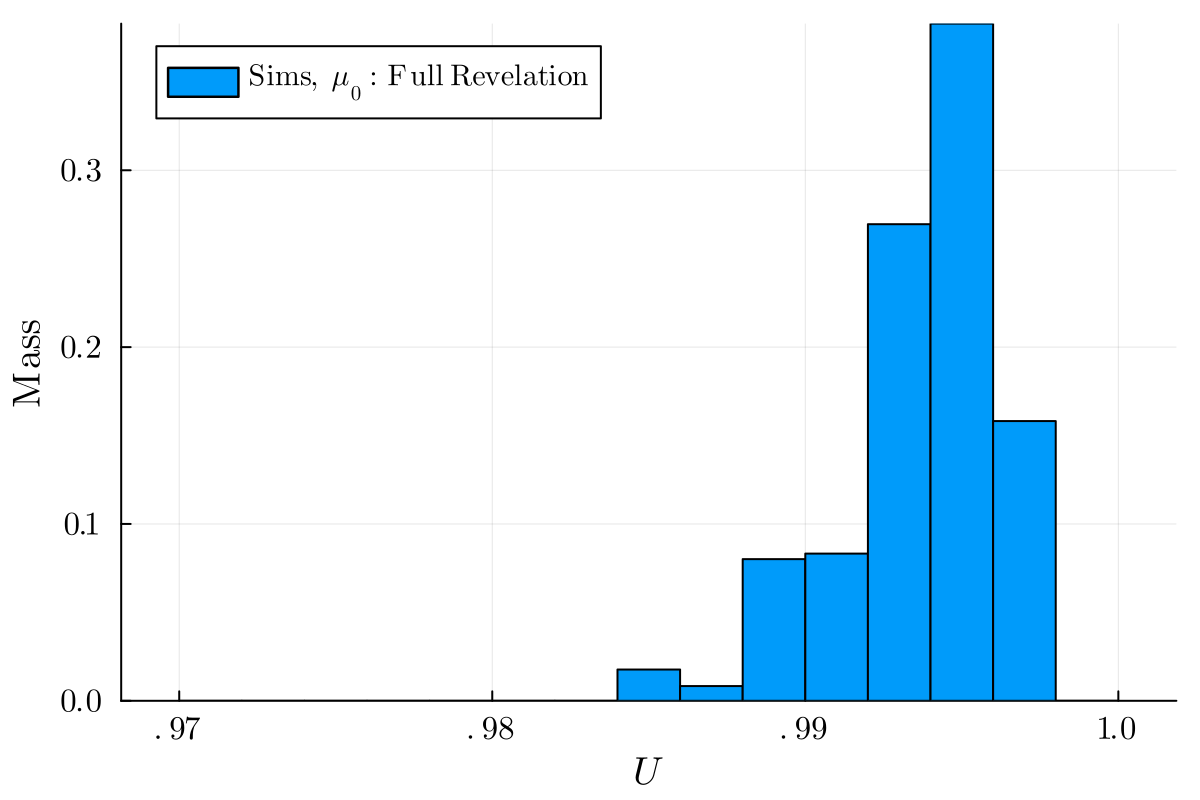}
  \caption{\footnotesize}
  \label{fig: nobias_heatfr}
\end{subfigure}
\caption{\footnotesize Normalized welfare of final policy in simulation for differing initial policies $\mu_0$ given $K=21$. The computation finds all pure equilibria featuring connected pools, i.e. if states $x,x'$ send message $m$, so will any $x'' \in [x,x']$. Any more general set of equilibria is computationally infeasible. (a) shows the welfare measure of final policies due to babbling initial policies, while (b) shows those due to fully revealing initial policies.  $T=10^8$, $N=1000$. The minimum and maximum level achieved are $.980$, $.995$ under babbling initial, and $.985, .997$ under full revelation initial. Here, $\underline U_K = .980$. }
\label{fig: nobias_heat}
\end{figure}
As predicted by the theory, even if the algorithm is started from a fully babbling initial strategy, it eventually learns highly informative final messaging strategies. Initializing with a fully revealing policy leads to consistently higher welfare levels than under babbling, but full revelation (a welfare equal to $1$) is never observed. We can use our constructive proof for the lower bound to calculate the threshold $\underline{U}_K$ for the case of our simulation experiments. In the environment underlying Figure \ref{fig: nobias_heat}, the bound is 98\% of the available utility which is exactly equal to the lowest value from the simulations.


Reward-driven adaptation comes with strings attached: while babbling always unravels, even initializations of full revelation lead to final policies featuring large pools close to the bounds of the state space, such as shown in Figure \ref{fig: nobias_polfr}.  Somewhat counterintuitively, the same reward-driven logic that unravels babbling may also curb the algorithm from learning to play the most informative equilibrium.\footnote{Figure \ref{fig: nobias_polbab} shows a typical final policy when the algorithm is initialized with babbling. } 
\begin{figure}
\centering
\begin{subfigure}{.5\textwidth}
  \centering
  \includegraphics[width=.9\linewidth]{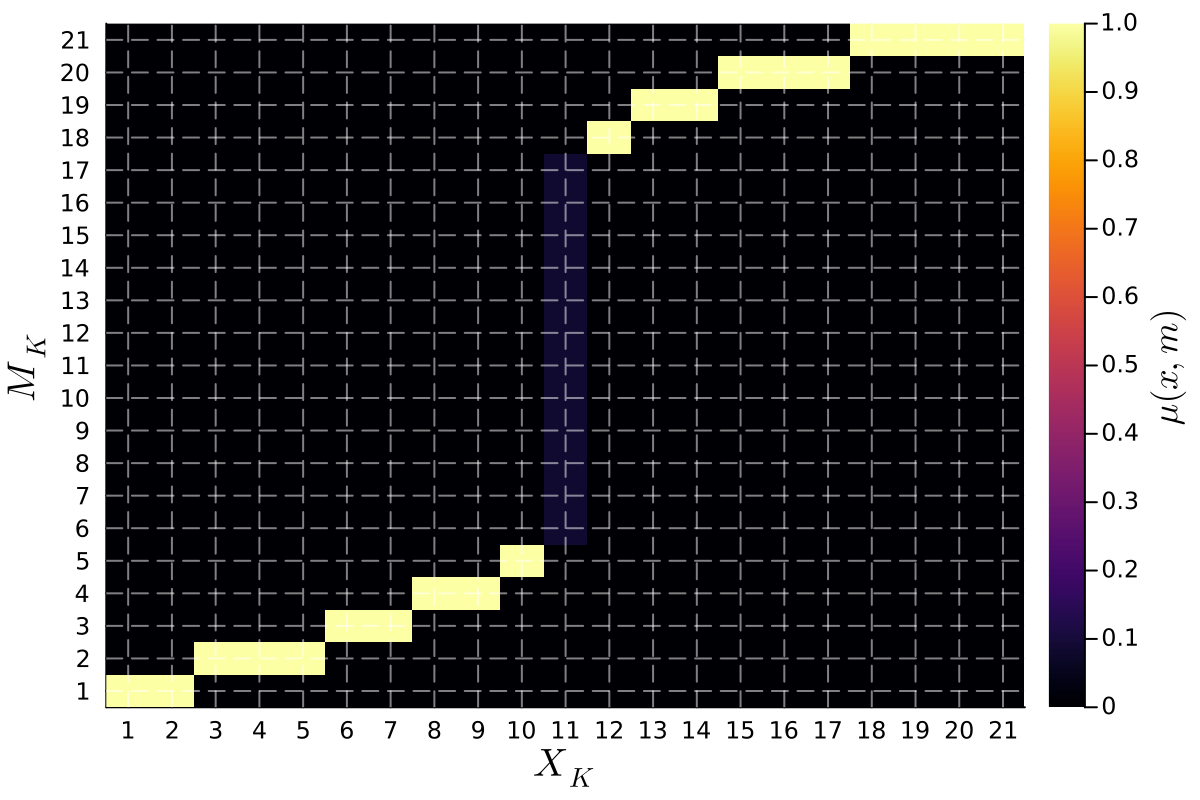}
  \caption{\footnotesize }
  \label{fig: nobias_polbab}
\end{subfigure}%
\begin{subfigure}{.5\textwidth}
  \centering
  \includegraphics[width=.9\linewidth]{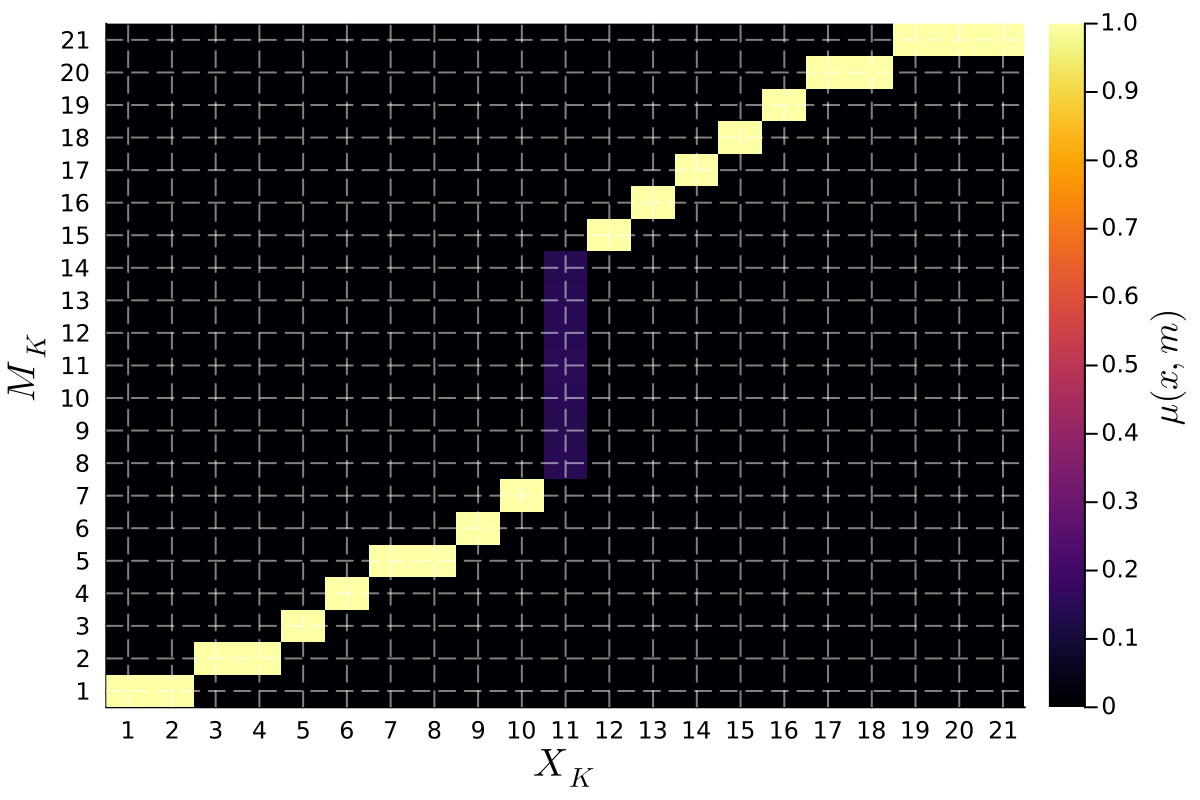}
  \caption{\footnotesize}
  \label{fig: nobias_polfr}
\end{subfigure}
\caption{\footnotesize Example of final policy learnt: we show the policy that achieves the median welfare level by normalized measure among our simulations. Here, (a) refers to the babbling initialization with median level of $0.986$, and (b) the fully revealing initialization with median level of $0.995$.  $T=10^8$, $N=1000$.}
\label{fig: nobias_pol}
\end{figure}

To see why full revelation easily unravels in the softmax setting, consider a sender who starts from a fully revealing policy and, upon observing state $x<<1/2$, sends the corresponding message. Because the sender uniformly experiments with non-negligible probability, the receiver knows that this message may also be sent in other states. Thus, after observing this message, the receiver does not attach probability one to the event that the state is $x$, but instead forms a posterior that is pulled toward the prior and, more precisely, toward those other states that occasionally send the same message. The induced action is therefore more cautious, in the sense that it lies strictly between $x$ and the prior mean $\tfrac{1}{2}$. As long as this effect is strong enough, the sender will learn that for states away from the prior, they may benefit from exaggerating,  that is, to induce the receiver to believe that the state is more extreme than it actually is by choosing the truthful message of a lower state $x'<x$. The resulting gain from exaggeration reinforces the use of such messages, which in turn pushes the receiver’s best response even closer to the prior. In equilibrium, this feedback loop generates pooling patterns in which larger pools appear near the boundaries of the state space, in a manner reminiscent of the classic \cite{crawford1982strategic} analysis. This is the pattern we see in Figure \ref{fig: nobias_pol}.\footnote{Even though $\varepsilon$ in our experiments is negligible, the algorithm initially starts with high levels of near-uniform exploration. Following the forces outlined here, during the periods of high exploration it may learn an equilibrium as shown in Figure \ref{fig: nobias_pol}. This is also an equilibrium for lower levels of exploration, and hence the learner may get stuck at that policy.}

The good news is that if one takes $\tau$ to a very low $\ut$ very slowly, the algorithm nearly always escapes these bad equilibria in our simulations and finds the fully revealing policy which is also a stable perfect Bayesian equilibrium of the learning dynamic. However, we will see in the next section that slow decay rates of $\tau$ lead to bad welfare outcomes when the sender's and the receiver's interests are \textit{not} aligned. Hence, even if the algorithm's designer is very patient, and does not care about the welfare cost associated with the excessive experimentation, they may still want to choose faster rates of decay if they are unsure about the existence of a preference bias.

\section{Positive Bias}

Consider now the case in which the incentives of the sender and the receiver are misaligned ($b \neq 0$). In this environment, the static Crawford–Sobel game admits only coarse communication: full revelation is no longer an equilibrium while babbling still is, and the set of equilibria consists of a finite number of ordered partitions of the state space. We find that the bias introduces new forces into the learning problem, modifies the shape of the limiting pooling structure, and fundamentally changes how exploration interacts with incentives.

Our first finding is that under reward-driven learning, the Crawford–Sobel equilibria lose their predictive power. To this end, define $\underline{b}_K = \frac{1}{2(K-1)}>0$.

    \begin{theorem}\label{thm: Posbiasnolearn}
        Let $b>\underline{b}_K$.
 For all $\zeta>0$ and $\ut>0, \varepsilon>0$ small enough, for all $\mu_0$, all $\mu \in \mathbf{NE}_{CS}$:  $\mathbb{P}\left[\mathbf{L}(\mu_0) \subseteq \mathbf{B}_{\zeta}(\mu) \right]=0$.
        \end{theorem}
Due to discreteness, only bias levels larger than a positive threshold will change incentives from the no-bias case to this case. This threshold $\underline{b}_K$ ensures at least half a stepsize between each discrete state. The theorem therefore rules out the uninteresting case $b< \underline{b}_K$. 

 Theorem \ref{thm: Posbiasnolearn} together with Theorem \ref{thm: Nashconv} implies that all learnt outcomes must indefinitely vary over time. This follows since if learning converges to a single strategy for the sender, then it must result in a PBE.  Due to the Theorem above, none of these equilibria are limit outcomes of the game with a positive bias.  

Exploration ensures that the receiver never interprets messages exactly as prescribed by any fixed pooling equilibrium, and the sender continually adjusts messages in response to this perceived distortion. As a consequence, the learning dynamics cannot settle. 

To gain traction on the problem, we perform numerical analysis via simulations using the same parametrization as in the previous section except we now set $b\in\{0.1,0.2\}$. 

Not surprisingly, all of our simulations for different levels of bias end with cycling behavior. To visualize this in Figure \ref{fig: bias_cyc}, we show two example simulations, where we track the last 500 iterations in which the policy changed by more than $0.2$ in at least one entry, and take the average over those iterations. The strategy in the figure is characterized by one or more stable pools at the top of the state space. After these pools there is a cutoff and below this cutoff the policy is cycling: Each type strictly below the cutoff learns to  imitate at least one type above it. Those types escape to playing actions used by some of the types above them. Eventually, the message that was initially used by the lowest type becomes practically only used due to the small baseline uniform experimentation $\varepsilon$. Hence, the receiver associates this message more and more with their prior belief. When types close to the middle learn this, they will start playing the message and hence escape any lower types they had joined with. This ``escalator'' pattern keeps rotating in a stable way.  

\begin{figure}
\centering
\begin{subfigure}{.5\textwidth}
  \centering
  \includegraphics[width=.9\linewidth]{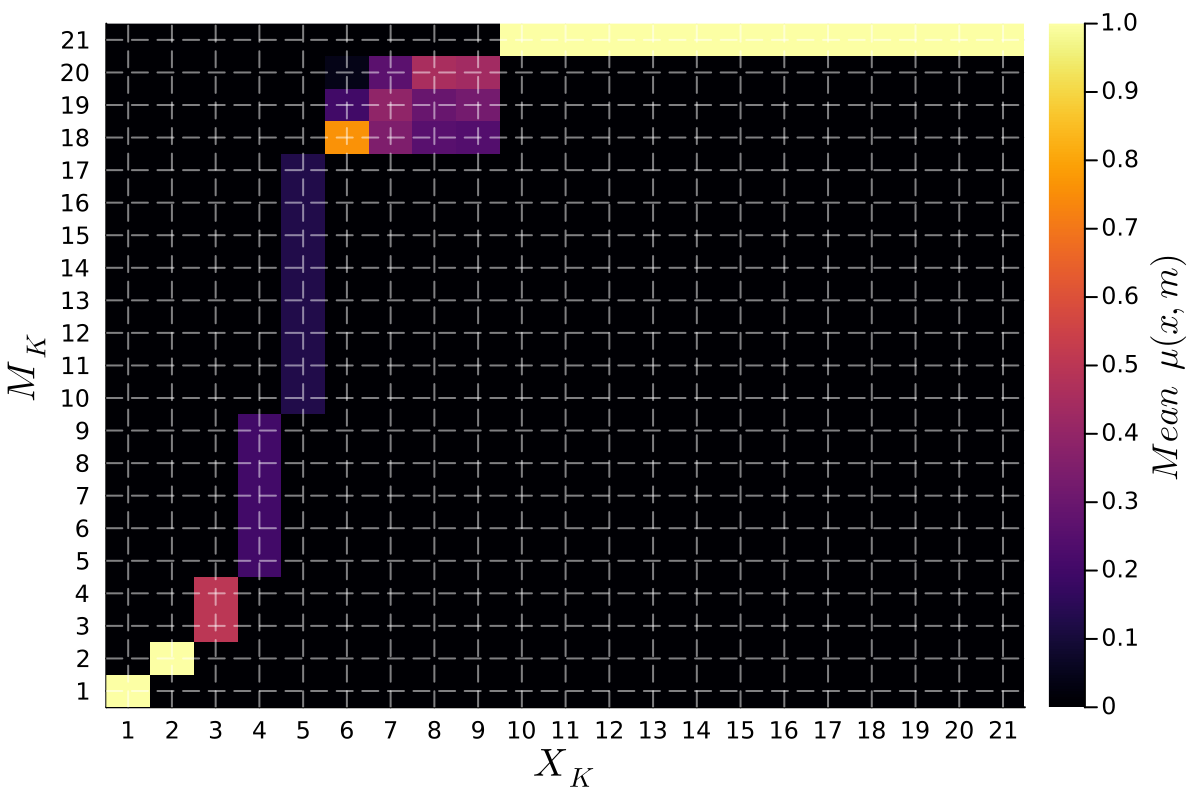}
  \caption{\footnotesize }
  \label{fig: bias_cycb1}
\end{subfigure}%
\begin{subfigure}{.5\textwidth}
  \centering
  \includegraphics[width=.9\linewidth]{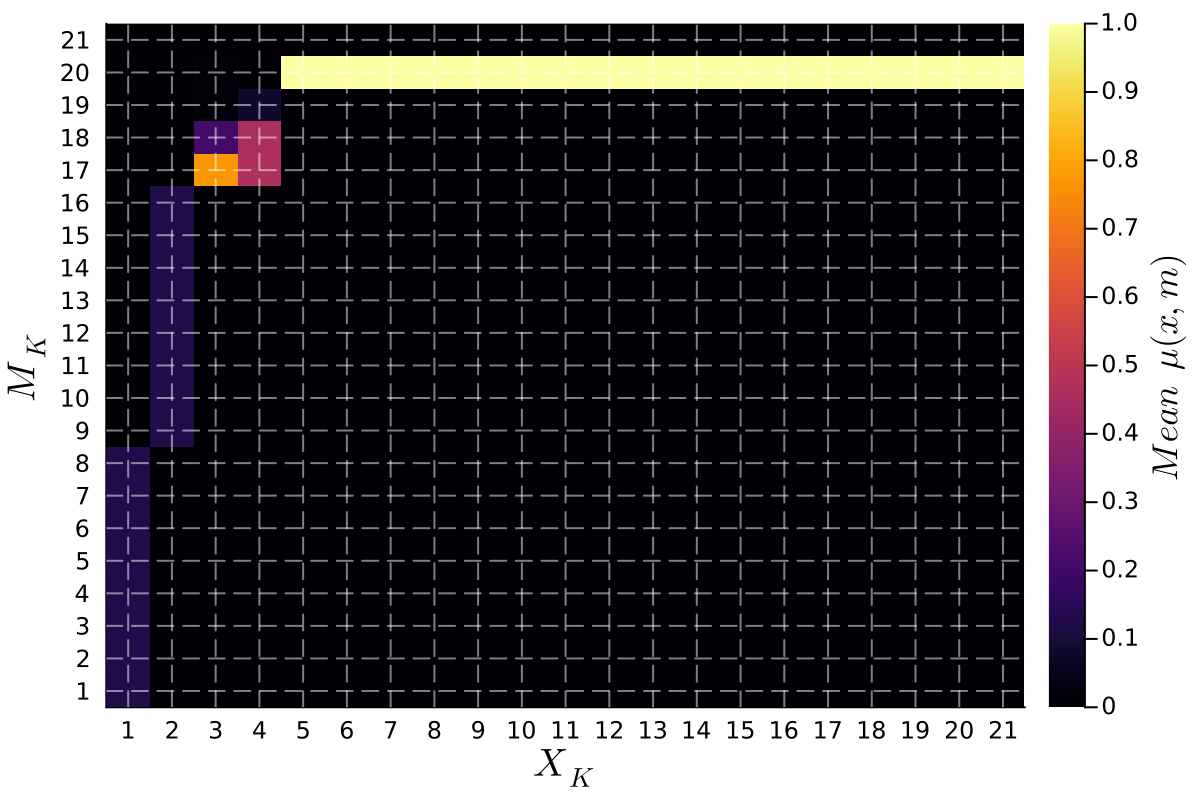}
  \caption{\footnotesize}
  \label{fig: bias_cycb2}
\end{subfigure}
\caption{\footnotesize Example of cyclical policy for two bias levels. We track the last 500 iterations whenever a policy changed by more than $0.2$ in at least one $(x,m)$ entry, and take an average. (a) refers to bias level $0.1$, (b) refers to a bias of $0.2$.}
\label{fig: bias_cyc}
\end{figure}

\textit{How much information is transmitted?}

Cycling does not hinder information transmission. Even though the language employed is constantly changing, the outcome is closer to full revelation than the \textit{best} Crawford-Sobel equilibria and consequently yields considerably higher welfare to both the sender and the receiver. This result can be seen in Figure \ref{fig: avg payoffs} where we plot the ratio between average welfare of the players in our simulations for different levels of bias and their expected welfare in the associated best Crawford-Sobel equilibrium. 

\begin{figure}[h]
      \centering
  \includegraphics[scale =1]{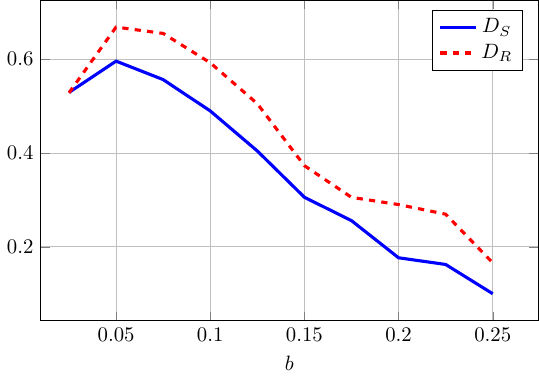}
  \caption{\footnotesize Payoff ratio $D_n = \frac{U^n-U^n_{CS}}{|U^n_{CS}|}$ for Sender and Receiver between payoff $U^n$ averaged over all final policies from simulation runs, and best PBE payoff $U^n_{CS}$, $n\in \{S,R\}$. The expected payoff is computed for the final policy of each simulation run, then averaged, given $N=1000$ simulations and $T=10^7$.}
  \label{fig: avg payoffs}
\end{figure}

A direct implication of this result is that when faced with a preference misalignment, if possible, an algorithm's designer like AirBnB has incentives to delegate the whole recommendation system to a learning algorithm. This is true even if the designer was able to solve for the best Crawford-Sobel equilibrium of the underlying game, as the algorithm's cycling communication strategy guarantees the sender a higher payoff than the best Crawford-Sobel equilibrium. Hence, in these types of communication problems delegating the recommendations to an algorithm not only simplify the designer's problem by freeing them from having to correctly estimate all model parameters but may also lead to better outcomes than what the designer can hope to reach without such commitment.

\begin{figure}[h]
      \centering
  \includegraphics[scale =1]{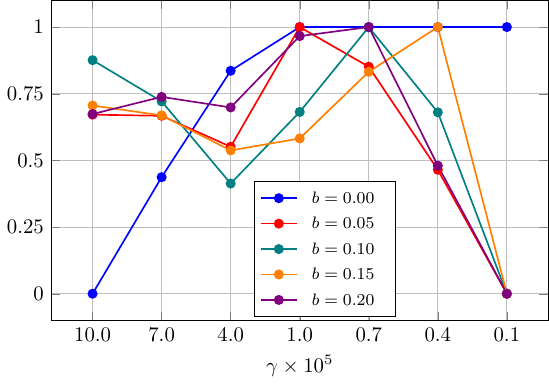}
  \caption{\footnotesize Average sender payoff over simulation runs for different exploration decay speeds $\gamma$ (x-axis), and different bias levels $b$ (series). To ease comparisons across bias levels, each series has been re-scaled to range from $0$ (minimal payoff among all series) to $1$ (maximal payoff among all series). The expected payoff is computed for the final policy of each simulation run, then averaged, given $N=1000$ simulations and $T=10^7$.}
  \label{fig: avpay_speeds}
\end{figure}

Finally, Figure \ref{fig: avpay_speeds} shows how the sender's payoff evolves in our simulations as we decrease the decay of $\tau$. In a striking contrast to the case of fully aligned interests, the final cycling policy eventually becomes increasingly worse for the sender, the slower the decay becomes. Hence, even a very patient designer of an algorithm who is not sure about the existence of a preference misalignment may want to use a faster decay speed than what would guarantee full revelation when preferences are fully aligned. In other words, the critical trade-off is not necessarily only about how much the designer loses because of an overly long experimentation period but also how much is lost due to convergence to a bad policy in case it turned out that the sender's and the receiver's preferences were in fact misaligned.

\section{Vocabulary selection and robustness}

An additional insight, relevant for applications, concerns the choice of vocabulary. For any given long-run communication pattern, there typically exist many different policies (i.e. mappings from states to messages) that could support it, and the assumption that the receiver fully discerns between all such vocabularies is in some sense strong. In applications there is usually some `natural' vocabulary, i.e. a predetermined one-to-one map between states and messages. This could be easily captured in the above framework by assuming that the natural language for $x$ is $m=x$ and by adding either some small cost for deviating from this association or assuming that a fraction of receivers are naive in the sense that they always interpret messages at face value. 

Across these variations, the same qualitative learning patterns re-emerged in our simulations, and the dynamics further selected the natural vocabulary. This suggests that in applications, the learning process itself helps anchor a simple and interpretable set of messages, making the overall communication scheme easier for naive receivers to learn and understand.

\section{Conclusion}

Taken together, our results show that reward-driven learning generates robustly informative communication patterns even in environments where most static cheap-talk equilibria are either less informative or fail to predict behavior. In the no-bias case, the algorithm reliably converges to highly informative pooling structures, while in the biased case it never settles on a Crawford–Sobel equilibrium and instead continues to adjust messages over time. In practical terms, our results imply that digital platforms that rely on Q-learning–type algorithms are naturally protected against extremely uninformative communication outcomes.

We  see several interesting avenues for future research. Firstly, this paper's analysis considers the setting where the receiver learns to best respond faster than the algorithmic sender; it is natural to ask how outcomes may change if the receiver's adaption instead lags behind the sender's. We conjecture that such a setting may lead to less informative outcomes. As the sender (algorithm) would be optimizing against an `outdated' receiver policy, it would likely over-correct its messaging strategy to influence a behavior that has already shifted. This could lead to miscoordinated dynamic where the sender's policy becomes increasingly extreme, potentially causing the `escalator' cycles we observe under bias to expand or collapse into babbling as the receiver's slow updates fail to anchor on a meaningful signal.

Furthermore, we are curious about how the theoretical results of this paper fare in applied settings. As a next step, we are considering behavioral experiments to verify our claims and investigate whether human decision-makers perceive and adapt to the potentially cyclical `language' of algorithmic advice.

\clearpage
\nocite{*}
\bibliographystyle{ACM-Reference-Format}
\bibliography{ct_refs}

\appendix

\section{Proofs}
\subsection{Proof of Theorem \ref{thm: Nashconv}}\label{app: nashconv}
We prove a more general result that will be useful for subsequent results. First, note that to be more precise for $\alpha_t$ in \eqref{eq: Qlearn}, we have that stepsizes for each pair $(x,m)$ advance only if $(x,m)$ is \textit{visited}. Hence one must introduce $n_t(x,m)$, the visitation count of each pair $x,m$, and write $\alpha_{n_t(x,m)}$ as the effective stepsize for each entry of the $Q$-matrix. Then we can write \eqref{eq: Qlearn} more concisely as
\begin{align}\label{eq: Qlearnprecise}
    Q_{t+1}(x,m) &= Q_t(x,m) + \alpha_{n_t(x,m)}\mathbf{1}_{t,x,m}\left(U^S(y(m,Q_t),x,b) - Q_t(x,m) + R_t\right),
\end{align}
where $\mathbf{1}_{t,x,m}$ is the indicator function for the visit $(x_t,m_t)=(x,m)$, and $R_t = u_t - U^S(y(m,Q_t),x,b)$ is the estimation error. Consider the following ordinary differential equation on $Q$-values:
 \begin{align}\label{eq: ODE}
    \dot{Q}(x,m) = \bar \mu(x,m)\Big(u(x,m,Q) - Q(x,m) \Big).
\end{align}
Define 
\begin{align*}
    \mathcal{Q}(\mu) = \left\{Q: \mu(x,m,Q)=\mu^*(x,m) \text{ for all } (x,m) \in \mathbf{X}_K\times \mathbf{M}_K \right\}, 
\end{align*}
the set of $Q$-values consistent with policy $\mu^*$. Then:
\begin{proposition}\label{prop: stochapproxQ}
    \begin{enumerate}
    \item[]
        \item For all $\mu_0$, $\mathbb{P}[\mathbf{L}(\mu_0) =  \{\mu^*\}]>0$ if  $\mathcal{Q}(\mu^*)$ is an attractor under \eqref{eq: ODE}.
        \item For all $\mu_0$, $\mathbb{P}[\mathbf{L}(\mu_0) =  \{\mu^*\}]=0$ if  $\mathcal{Q}(\mu^*)$ is linearly unstable under \eqref{eq: ODE}.
\item If $Q$ is a rest point of \eqref{eq: ODE}, then for all $\zeta>0$ there is $\ut>0, \eps>0$ small enough, the policy $\mu$ induced by $Q$ satisfies $\mu \in \mathbf{B}_{\zeta}(\mu^*)$ for some $\mu^* \in \mathbf{NE}_{CS}$.
    \end{enumerate}
\end{proposition}
\begin{proof}
    
For point (1), a standard stochastic approximation argument gives \eqref{eq: ODE} as the limiting ODE describing long-run $Q_t$ when following \eqref{eq: Qlearnprecise}. See e.g. \cite[Chapter 7]{borkar2008stochastic}, which accounts for the asynchronous nature of $Q$-learning by introducing the scaling factors $\bar \mu(x,m)$ in \eqref{eq: ODE}. The result then follows by a standard argument: see \cite[Theorem 7.3]{benaim2006dynamics}. For the partial converse point (2), by standard arguments as in \cite[Theorem 3.12]{benaim2012stochastic} it is sufficient to have (see \cite[Hypothesis 3.11]{benaim2012stochastic}) the presence of shocks $\eta$ to payoff realizations $u_t$ of the sender, as introduced in the model section. Intuitively, unstable equilibria have zero probability to be learnt as long as errors ($\eta$) in the updates of the algorithm persistently push it towards a repelling direction of the equilibrium. 

For point (3), consider the simplified ODE over $Q$-matrices we refer to as $W(x,m)$, to keep track of the fact that these are generated by different dynamics: 
\begin{align}\label{eq: ODEsimpl}
     \dot{W}(x,m) = \Big(u(x,m,W) - W(x,m) \Big).
\end{align}
Note that, as $\bar \mu(x,m)>0$ always, zeros of \eqref{eq: ODEsimpl} and \eqref{eq: ODE} coincide. Furthermore, $\dot W(x,m)$ induces the following dynamics on $\mu(x,m)$:
\begin{align*}
    \dot \mu(x,m,W) = \tau^{-1}\mu(x,m,W)\Big[\dot W(x,m) - \sum_{m'}\mu(x,m',W)\dot W(x,m') \Big],
\end{align*}
as can be derived from the definition of $\mu$. Furthermore, as
\begin{align*}
    \ln(\mu(x,m,W)) = \tau^{-1}W(x,m) - \ln\Bigg(\sum_{m'}\exp\Big(\tau^{-1}W(x,m') \Big) \Bigg) = \tau^{-1}W(x,m) - \mathbf{L}(W),
\end{align*}
we get
\begin{align*}
   \tau\dot \mu(x,m,W)  &= \mu(x,m,W)\Big(u(x,m,W) - \sum_{m'}\mu(x,m',W) u(x,m',W) \\
   &- \tau \Big[\ln\big(\mu(x,m,W)\big) - \sum_{m'}\mu(x,m',W)\ln\big(\mu(x,m',W)\big)  \Big]\Big).
\end{align*}
We can drop $\tau$ on the left hand side as it just re-scales time. As $u(x,m,W)$ depends on $X$ only through $\mu$, we can more concisely write
\begin{align}
\begin{split}
 \dot \mu(x,m) &= \mu(x,m)\Big(u(x,m,\mu) - \sum_{m'}\mu(x,m')u(x,m',\mu)\\
 &- \tau \Big[\ln\big(\mu(x,m)\big) - \sum_{m'}\mu(x,m') \ln\big(\mu(x,m')\big)  \Big]\Big)\\
 &= \mu(x,m)\left(u^H(x,m,\mu,\tau) - \sum_{m'}\mu(x,m')u^H(x,m',\mu,\tau)\right),\label{eq: piODE}
 \end{split}
\end{align}
where 
\begin{align*}
        u^H(x,m,\mu,\tau) = u(x,m,W) - \tau \log(\mu(x,m,W)),
    \end{align*}
    are perturbed payoffs.
As $\bar \mu>0$ always, rest points of \eqref{eq: ODEsimpl} (and therefore of \eqref{eq: ODE}) are rest points of \eqref{eq: piODE}. From \eqref{eq: piODE}, it follows that rest points of \eqref{eq: ODE} are Nash equilibria in the perturbed game with payoffs $u^H$.

Let $\mathbf{NE}_{\eps}(\ut)$ be the set of (sender policies $\mu$ in) such Nash equilibria. We need to show that the Nash equilibrium correspondence $\mathbf{NE}_{\varepsilon}(\ut)$ is upper hemicontinuous at $\varepsilon,\ut =0$. Take an equilibrium $\mu_{\varepsilon,\ut} \in \mathbf{NE}_{\eps}(\ut)$. By contradiction, suppose there exists a sequence $(\eps_t,\ut_t)_{t=1}^{\infty}$ s.t. $\mu_{\eps_t,\ut_t}$ are equilibria for every $t$, but $\lim_{t\rightarrow \infty}\mu_{\eps_t,\ut_t}\notin \mathbf{NE}_0(0)=\mathbf{NE}_{CS}$. Equivalently, there exists $x,m,m'$ such that 
\begin{align*}
    U^S(x,m,\mu_{0,0}) < U^S(x,m',\mu_{0,0}),
\end{align*}
but $\mu_{0,0}(x,m')=0$, as would happen for non-argmax messages given $\ut=0$. By continuity of $\mu$ in $\eps, \ut$, and continuity of sender payoffs $U^S$ in $\mu$, we must have that this strict inequality is preserved also for $\eps>0, \ut>0$ small enough. But then we have a contradiction: for such $\eps, \ut$, $m'$ would be the argmax message given $x$, but isn't sent with maximal probability under $\mu_{\eps,\ut}$. Thus, there exists $t $ large enough so that $\mu_{\eps_t,\ut_t} \notin \mathbf{NE}_{\eps_t}(\ut_t)$, which contradicts the setup. The result follows. 
\end{proof}
Point (3) of the above Proposition implies the conclusion of Theorem \ref{thm: Nashconv}: $\mathbf{L}(\mu_0)$ being a single point implies convergence, and convergence can only happen if a rest point of \eqref{eq: ODE} is found. Then apply point (3). \hfill $\square$

\subsection{Proof of Theorem \ref{thm: attrchar}}\label{app: attrcharproof}
The proof involves multiple steps. 

Given the conclusion of Proposition \ref{prop: stochapproxQ}, we can move to showing that MSFR and SAPS are characterising features of attracting Nash equilibria under \eqref{eq: ODE}. First, a few definitions. 

For any state $x\in \mathbf{X}_K$, and fixing a policy $\mu$, let $G(x)$ be the set of messages that maximize $\mu(x,\cdot)$. Now note that 
\begin{align}\label{eq: pi deriv}
    \frac{\partial \mu(x',m)}{\partial Q(x',m')} = \begin{cases}
             \ut^{-1}\mu(x',m)\big(1 - \mu(x',m) \big)  &\text{ if } m' = m\\
            -\ut^{-1}\mu(x',m)\mu(x',m') & \text{ o.w. }
        \end{cases},
\end{align}
which follows from the definition of the softmax. 


\textbf{If}: For the `if' part of the characterisation of attracting equilibria, it will be helpful to note that for $\eps,\ut $ small enough, $U(Q)$, the vector of $u(x,m,Q)$ over $x,m$  is a local contraction around $Q^*(\mu)$, the expected payoffs of the equilibrium $\mu$. To see this, consider
\begin{align*}
         \frac{\partial u(x,m,Q)}{\partial Q(x',m')} &= 2\big(x - y(m,Q) \big)(1-\eps)\frac{\partial \mu(x',m)}{\partial Q(x',m')}\Bigg[ \frac{\Big(x' - \mathbb{E}[x\mid m,\mu]\Big)}{(1-\eps)\mathbb{P}[m \mid \mu] +\eps}  \\
         &-  \frac{\eps \Big(\frac{1}{2} - \mathbb{E}[x \mid m, \mu] \Big)}{((1-\eps)\mathbb{P}[m \mid \mu]+ \eps)^2} \Bigg]\\
         &= \frac{\partial \mu(x',m)}{\partial Q(x',m')}A(x,x',m).
    \end{align*}
    From \eqref{eq: pi deriv} we see that $\frac{\partial \mu(x',m)}{\partial Q(x',m')} = o\big(\exp(-\ut^{-1})\big)$ whenever $|G(x')|=1$, i.e. for states $x'$ that send a single message with high probability. Furthermore, note that when the middle state $x=\frac{1}{2}$ is fully revealing, no matter whether they mix over messages or not, $\frac{\partial u(x,m,Q)}{\partial Q(x',m')}=0$ for all $m,x',m'$. Now consider a $\mu \in \mathbf{NE}_{CS}$ (by upper hemicontinuity as shown in the proof of Theorem \ref{thm: Nashconv}, we can consider $\mathbf{NE}_{CS}$ given $\ut, \eps$ small enough)  that satisfies MSFR and SAPS. Let $\omega(x)$ be the (set of) messages with highest $\mu(x,m)$ given $\mu$.  Define the optimality gap as
    \begin{align*}
        \Delta = \min_{x\in \mathbf{X}_K}\left( u(x,\omega(x),Q^*) - \max_{m' \notin \omega(x)}u(x,m',Q^*) \right),
    \end{align*}
    where with some abuse of notation $\omega(x)$ selects an arbitrary element of the set as payoffs are equal on the set. For some $c\in (0,1)$, let
    \begin{align*}
         N_c = \left\{Q: \| Q - U(Q^*)\|_{\infty}< c\Delta \right\}.
    \end{align*}
    as the set of all $Q$ that are close enough to $Q^*$ so that for all $x'\neq 1/2$, optimal messages will match with $Q^*$, and for $x=1/2$ the set of messages sent (almost) only by $1/2$ equals $\omega(1/2)$, for $\ut$ small enough. We will show that $U(Q)$ is a contraction on $N_c$. Clearly, from the above we have that $\left| \frac{\partial u(x,m,Q)}{\partial Q(x',m')}\right| = o\big(\exp(-\ut^{-1})\big)$ for all $x\neq 1/2$. For $x = 1/2$, we have that $\mu(x,m,Q)$ is not generically constant for all $m\in \omega(x)$. However, for all $\ut>0$ we can find $c>0$ small enough so that
    \begin{align*}
        \left|\frac{\partial u(1/2,m,Q)}{\partial Q(x',m')}\right| <1,
    \end{align*}
    for all $m, m' \in \omega(x)$. Note that $c>0$ controls the size of $x-y(m,\mu(Q))$, and hence $A(1/2,x',m,Q)$: the smaller $c>0$, the more likely it is that almost only $x$ sent message $m$, and therefore the posterior given $m$ approaches $x$. Given fixed  $\ut$,  fix such $c>0$ small enough and consider the neighborhood $N_c$.

     Let $0<\beta<1$ be the parameters such that whenever $Q,Q' \in N_c$,
    \begin{align*}
        \|U(Q) - U(Q')\|\le \beta\|Q-Q'\|.
    \end{align*}
    
    Now we can follow the approach taken in \cite[Chapters 7.4 and 10.3]{borkar2008stochastic}. First, we can write \eqref{eq: ODE} equivalently as
    \begin{align}\label{eq: odeinclu}
        \dot Q \in \Omega_{\eps}\cdot \Big[U(Q) - Q \Big],
    \end{align}
    where $\cdot$ indicates element-wise multiplication, $\Omega_{\eps}= [\eps,1]^{K^2}$, and $U(Q)$ is the vector of $u(x,m,Q)$ over $x,m$. Thus, \eqref{eq: odeinclu} re-interprets \eqref{eq: ODE} as a differential inclusion with varying scaling parameter vectors $\omega \in \Omega_{\eps}$. Re-write this further as
    \begin{align*}
        \Omega_{\eps}\cdot \Big[U(Q) - Q \Big] = \overline{U}(Q) - Q,
    \end{align*}
    where $\overline{U}(Q) = (I_{K^2}-\Omega_{\eps})\cdot Q + \Omega_{\eps}\cdot U(Q)$, $I_{K^2}$ being the identity matrix of dimension $K^2$. It follows that for any  $Q,Q' \in  N_c$, 
    \begin{align*}
        \|\overline{U}(Q) - \overline{U}(Q')\|\le \overline{\beta}\|Q-Q'\|,
    \end{align*}
where $\overline{\beta} = 1-\eps(1-\beta) \in (0,1)$. Thus, $\overline{U}(Q)$ is a contraction whenever $U(Q)$ is. To finish the proof, we consider a Lyapunov function as in \cite[Chapter 10.3]{borkar2008stochastic}. Define, for any positive weights $w(x,m)>0$, for all $x,m$, the norm
\begin{align*}
    \|Q\|_{w,p} = \left( \sum_{x,m}w(x,m)|Q(x,m)|^p\right)^{\frac{1}{p}},
\end{align*}
with $1\le p < \infty$.  For any initial value $Q_0$,  $V(t) = \|Q(t) - Q^*\|_{w,p}$, $t\ge 0$. Then, using an argument analogous to the proof of Theorem 2 in \cite[Chapter 10.3]{borkar2008stochastic}, we obtain
\begin{align*}
    \frac{d}{dt}V(t) \le \|\overline{U}(Q(t)) - \overline{U}(Q^*)\|_{w,p} - \|Q(t) - Q^*\|_{w,p} \le (\bar \beta -1)\|Q(t) - Q^*\|_{w,p}<0,
\end{align*}
whenever $Q(t)\in N_c$. The result follows: $V(t)$ must decrease to zero whenever $Q(t) \in N_c$, and the equilibrium is approached. Hence, any $Q^*$ equilibrium payoff associated with a $\mu \in \mathbf{NE}_{CS}$ satisfying MSFR and SAPS will be an attractor under \eqref{eq: ODE}, and therefore learnt with positive probability, as required by our `if' part. 

\textbf{Only if:} For the `only if' part, we show that a $\mu \in \mathbf{NE}_{CS}$ that does not sastify MSFR or SAPS will either not be feasible as equilibrium to be learnt under the softmax, or linearly unstable, and therefore learnt with zero probability. We will proceed via multiple Lemmas. First, consider the case of an equilibrium policy $\mu$ under which some states fully reveal by mixing over multiple messages. 

To check stability of any equilibrium $\mu$, we linearize \eqref{eq: ODE} at $\mu$. $Q^*$ be the $Q$-values associated with $\mu$. Since $Q$ is a $K\times K$ matrix, we can re-write the system as a vector system $\tilde{Q} \in \mathbb{R}^{K^2}$. The linearization is then the Jacobian $J(\tilde{Q^*})\in \mathbb{R}^{K^2\times K^2}$, which for short we write as $J$. We normalize the linearized system so that $Q^*$ lies at the origin. 

We prove a more general result that will be useful later on: any equilibrium policy $\mu$ that involves mixtures over messages for pools other than the singular pool $P=\{1/2\}$ must be unstable, and therefore will converged to with probability zero.

\begin{lemma}\label{lem: nomixpools}
    Consider $\mu \in \mathbf{NE}_{CS}$ with a set of states $P\subseteq X_K$ with $|P|=L\ge1$, s.t. there are messages $m_1,m_2 \in G(P)$. Then whenever $P\neq \{1/2\}$, for all $\ut, \varepsilon$ small enough, $\mu$ is linearly unstable under \eqref{eq: ODE}.
\end{lemma}
\begin{proof}
Note that as $\mu$ is an equilibrium, by the nature of the softmax policy it must be that all states in $P$ send the same set of messages, and all messages in that set with equal probability. Otherwise, there would be some message $m\neq m_1, m_2$ sent by one state $x\in P$ that is not sent by some other state $x'\in P$. As mixing in equilibrium only happens when states are indifferent over messages, and both $x,x'$ are sending $m_1,m_2$ with positive probability, $x,x'$ must also be indifferent over $m$, and so must put positive mass on $m$ by construction of the softmax. Now, let $N>1$ be the number of messages sent by each $x_i \in P$. Let $L=|P|\ge1$. Let $E = \mathbb{E}[x\mid m_j,\mu]$, and $\overline{E} = y(m_j,Q)$ for $m_j\in G(P)$. Define $O(m)$ as the set of states $x$ for which $m\in G(x)$. Finally, write $C =K\mathbb{P}(m_j|\mu) = \sum_{i: x_i \in O(m_j)}\frac{1}{N}=\frac{L}{N}$ as the probability that $m_j\in G(P)$ is sent under $\mu$, scaled by $K$. Write $\nu = \exp(-\ut^{-1})$.  For $m \in G(P)$ define
\begin{align*}
    A(x,x',m) = 2\big(x - \overline{E} \big)(1-\eps)\Bigg[ \frac{\Big(x' - E\Big)}{(1-\eps)C +\eps}  -  \frac{\eps \Big(\frac{1}{2} - E\Big)}{((1-\eps)C+ \eps)^2} \Bigg] + o(\nu),
    \end{align*}
     Then:
    \begin{align*}
        \frac{\partial u(x,m,Q)}{\partial Q(x_i,m')} &= \begin{cases}
            \ut^{-1}\frac{N-1}{N^2}A(x,x_i,m) + o(\nu) & \text{ if } m' = m \in G(x_i) \\
            -\ut^{-1}\frac{1}{N^2}A(x,x_i,m) + o(\nu) & \text{ if } m' \neq m \text{ and }m,m' \in G(x_i)\\
            o(\nu) & \text{ o.w.}
        \end{cases}
    \end{align*}
     To construct an eigenvector, write $\beta(x,m)$ as the entry associated with $Q(x,m)$. Suppose $m_j \in G(P)$. Then
    \begin{align*}
        (J\beta)[x_i,m_j] &= &\mu(x_i,m_j)\left[\sum_{k=1}^K\sum_{k'=1}^K\beta(x_k,m_{k'})\frac{\partial u(x_i,m_j,Q)}{\partial Q(x_k,m_{k'})} - \beta(x_i,m_j)\right]\\
        &=& \mu(x_i,m_j)\Bigg[R(x_i,m_j)\sum_{x_j \in P}\Bigg[ \frac{\Big(x_j - E\Big)}{(1-\eps)C +\eps}  -  \frac{\eps \Big(\frac{1}{2} - E\Big)}{((1-\eps)C+ \eps)^2} \Bigg]\\
        && \Big[(N-1)\beta(x_j,m_j) - \sum_{m_{k'} \in G(P), k' \neq j}\beta(x_j,m_{k'})   \Big] - \beta(x_i,m_j)\Bigg] + o(\nu)
    \end{align*}
    where $R(x_i,m_j) = \ut^{-1}2\big(x_i - \overline{E}\big)(1-\eps)\frac{1}{N^2}$.  
    
 Without loss, let $\{m_1,\ldots,m_{N}\} = G(P)$. We consider a perturbation to Q values of these messages, for all states. Consider the following vector:
\begin{align*}
    \beta(x,m_1) &= - (N-1)\beta(x,m_2) & \forall x&\\
    \beta(x,m_2) &= \beta(x,m_j) &  2 \le j& \le N\\
    \beta(x,m_1)  &=  \frac{x - \bar{E}}{x' - \bar{E}}& \forall\; &x,x'\in P &\\
    \beta(x,m_j) &=0 & \text{ if } j> &N.
\end{align*}
We will proceed by showing that $\lambda\neq 0$ is an approximate eigenvalue of $J$. For this, consider $\tilde J$, the operator that differs from $J$ only by setting all $o(\nu) $ parts of $J$ to zero. First, note that if $\beta$ are such that there exists an eigenvalue $\lambda \neq 0$, then for all $m_j,m_{j'}\in G(P)$ and $x_i$ s.t. $\beta(x_i,m_j), \beta(x_i,m_{j'}) \neq 0$, we have:
\begin{align}\label{eq: posb_betratio}
\begin{split}
 \frac{\left(\mu(x_i,m_j)+\lambda\right)\beta(x_i,m_j)}{\left(\mu(x_i,m_{j'})+\lambda\right)\beta(x_i,m_{j'})} \\
=    \frac{\sum_{x_j \in P}\Bigg[ \frac{\Big(x_j - E\Big)}{(1-\eps)C +\eps}  -  \frac{\eps \Big(\frac{1}{2} - E\Big)}{((1-\eps)C+ \eps)^2} \Bigg]\Big[(N-1)\beta(x_j,m_j) - \sum_{m_{k'} \in G(P), k' \neq j}\beta(x_j,m_{k'})   \Big]}{\sum_{x_j \in P}\Bigg[ \frac{\Big(x_j - E\Big)}{(1-\eps)C +\eps}  -  \frac{\eps \Big(\frac{1}{2} - E\Big)}{((1-\eps)C+ \eps)^2} \Bigg]\Big[(N-1)\beta(x_j,m_{j'}) - \sum_{m_{k'} \in G(P), k' \neq j'}\beta(x_j,m_{k'})   \Big]} 
\end{split}
\end{align}
To check the properties of $\beta$ are consistent, we first consider the first property of the vector: 
 \begin{align*}
     \frac{\beta(x_i,m_1)}{\beta(x_i,m_{2})}&=\\
     \frac{\sum_{x_j \in P}\Bigg[ \frac{\Big(x_j - E\Big)}{(1-\eps)C +\eps}  -  \frac{\eps \Big(\frac{1}{2} - E\Big)}{((1-\eps)C+ \eps)^2} \Bigg]\Big[(N-1)\left(\beta(x_j,m_1) -\beta(x_j,m_2)\right)   \Big]}{\sum_{x_j \in P}\Bigg[ \frac{\Big(x_j - E\Big)}{(1-\eps)C +\eps}  -  \frac{\eps \Big(\frac{1}{2} - E\Big)}{((1-\eps)C+ \eps)^2} \Bigg]\Big[\beta(x_j,m_2) - \beta(x_j,m_1)   \Big]}\\
     &= -(N-1)
 \end{align*}
where we first use that for this policy, every state $x'$ sends all messages $m_j$, $1\le j\le N$, with  equal probability. Second we plug in Similarly one can verify that this is consistent with having $\beta(x,m_2) = \beta(x,m_j)$ for $2\le j\le N$.

 Note that when $x_i,x_{i'} \in P$, we have that $\frac{\mu(x_i,m_1)}{\mu(x_{i'},m_1)}=1$. Hence, for such pairs of states,
\begin{align*}
\frac{\beta(x_i,m_1)}{\beta(x_{i'},m_1)} &= \frac{x_i - \bar{E}_1}{x_{i'} - \bar{E}_1}   \frac{\sum_{x_j \in P}\Bigg[ \frac{\Big(x_j - E\Big)}{(1-\eps)C +\eps}  -  \frac{\eps \Big(\frac{1}{2} - E\Big)}{((1-\eps)C+ \eps)^2} \Bigg]\beta(x_j,m_1)}{\sum_{x_j \in P}\Bigg[ \frac{\Big(x_j - E\Big)}{(1-\eps)C +\eps}  -  \frac{\eps \Big(\frac{1}{2} - E\Big)}{((1-\eps)C+ \eps)^2} \Bigg]\beta(x_j,m_1)  } \\
& = \frac{x_i - \bar{E}}{x_{i'} - \bar{E}}, 
\end{align*}
as required. We can write
\begin{align*}
 &\sum_{x_j \in P}\Bigg[ \frac{\Big(x_j - E\Big)}{(1-\eps)C +\eps}  -  \frac{\eps \Big(\frac{1}{2} - E\Big)}{((1-\eps)C+ \eps)^2} \Bigg]\beta(x_j,m_1) \\
 &= \beta(x_1,m_1) \sum_{x_j \in P}\Bigg[ \frac{\Big(x_j - E\Big)}{(1-\eps)C +\eps}  -  \frac{\eps \Big(\frac{1}{2} - E\Big)}{((1-\eps)C+ \eps)^2} \Bigg]\frac{x_j - \bar{E}}{x_1 - \bar{E}}.
\end{align*}
To find $\lambda$, let $\beta(x_1,m_1) = 1$. Then, 
\begin{align*}
    1 + \frac{\lambda}{\mu(x_1,m_1)} &= NR(x_1,m_1)\sum_{x_j \in P}\Bigg[ \frac{\Big(x_j - E\Big)}{(1-\eps)C +\eps}  -  \frac{\eps \Big(\frac{1}{2} - E\Big)}{((1-\eps)C+ \eps)^2} \Bigg]\beta(x_j,m_1)\\
    &= NR(x_1,m_1)\sum_{x_j \in P}\Bigg[ \frac{\Big(x_j - E\Big)}{(1-\eps)C +\eps}  -  \frac{\eps \Big(\frac{1}{2} - E\Big)}{((1-\eps)C+ \eps)^2} \Bigg]\frac{x_j - \bar{E}}{x_1 - \bar{E}}\\
    &= \ut^{-1}2(1-\eps)\frac{1}{N}\sum_{x_j \in P}\Bigg[ \frac{\Big(x_j - E\Big)}{(1-\eps)C +\eps}  -  \frac{\eps \Big(\frac{1}{2} - E\Big)}{((1-\eps)C+ \eps)^2} \Bigg]\big( x_j - \bar{E}\big).
\end{align*}
For instability, we need
\begin{align*}
    \lambda &> 0\\
    \frac{1}{N}\sum_{j =1}^{N}\Bigg[ \frac{\Big(x_j - E\Big)}{(1-\eps)C +\eps}  -  \frac{\eps \Big(\frac{1}{2} - E\Big)}{((1-\eps)C+ \eps)^2} \Bigg]\big( x_j - \bar{E}\big) &> \frac{\ut }{2(1-\eps)}.
\end{align*}
Note that for $\eps$ small enough, the LHS is strictly positive for all possible pools $P$ except if $P=\{1/2\}$. In that case, for all $\ut>0$, this eigenvalue is negative, and we cannot conclude instability. When $P\neq \{1/2\}$, the left hand side approaches the conditional variance of $x$ given $m_1$. The RHS is $o(\ut )$, and thus the result follows.
\end{proof}

Hence, any equilibrium policy involving mixtures that are not only happening given state $1/2$ must be unstable, and cannot be learnt. 

The following Lemma verifies that unused the presence of unused messages will lead to a posterior arbitrarily close to the prior $1/2$, as $\ut$ shrinks. For a given policy $\mu$, define $\mathbf{M}_0\subset \mathbf{M}$ as the set of `unused' messages under $\mu$, i.e. the messages not sent as a maximizer of $\mu$ for any state. 
\begin{lemma}\label{lem: unusedM}
    Take $\mu$ such that $\mathbf{M}_0\neq \emptyset$. For all $m\in \mathbf{M}_0$, $y(m,\bar\mu) = \frac{1}{2} + o\left(\nu/\eps\right).$
\end{lemma}
\begin{proof}
For any $m\in \mathbf{M}_0$, $x$, we have $\bar \mu(x,m)= \frac{\eps}{K} + (1-\eps)\mu(x,m) = \frac{\eps}{K} + o(\nu)$. Then, 
    \begin{align*}
        y(m,\bar \mu) &= \mathbb{E}[x\mid m,\bar \mu] = \frac{\frac{1}{2}\frac{\eps}{K} + o(\nu) }{\frac{\eps}{K}  + o(\nu)}\\
        &= \frac{1}{2} + o(\nu/\eps).
    \end{align*}
\end{proof}

Hence, as long as $\nu/\eps$ vanishes, any unused message will get a receiver action arbitrarily close to $1/2$. Recalling that $\nu = \exp(-\ut^{-1})$, this is satisfied as long as $\eps$ vanishes slower than exponentially. 



This concludes that MSFR is necessary. Now for the necessity of SAPS. We show that if SAPS fails, the softmax cannot support such an equilibrium. When SAPS fails, we must have two neighboring pools $P,P'$ with $|P|=M,|P'|=M-2$, so that $m_1$ sent by states in $P$, $m_2$ by states in $P'$. Without loss, all states are smaller than $1\over 2$. Define $\bar x = \max_{x\in P}x$, and suppose $\min_{x\in P'}x>\bar x$.
\begin{lemma}\label{lem: SAPS}
    For all $\delta \in (0,\frac{1}{2})$, there exists $\eps>0,\ut>0$, small enough so that $\mu(\bar x,m_1)\le \delta$, for all $\mu\in \mathbf{NE}_{CS}$ featuring $P,P'$ as defined above.
\end{lemma}
\begin{proof}
Define $E_i = \mathbb{E}[x\mid m_i,\mu]$ as the conditional mean of each pool, disregarding uniform exploration $\eps$. Note that by definition, $\bar E_1  =\bar x - \frac{M-1}{2(K-1)}\frac{M}{M-1+\mu}$, and $\bar E_2 = \bar x + \frac{M-1}{2(K-1)}\frac{M-2}{M-1-\mu}$, where $\mu=\mu(\bar x,m_1)$. Then, 
    \begin{align*}
        v_1=\Big(\bar x - y(m_1) \Big)^2 &= \Bigg(\frac{M-1 }{2(K-1)}\frac{M}{M-1+\mu} - \frac{\eps(\frac{1}{2}-\bar E_1)}{(1-\eps)(M-1+\mu) + \eps} \Bigg)^2,\\
        v_2=\Big(\bar x - y(m_2) \Big)^2 &= \Bigg(\frac{M-1 }{2(K-1)}\frac{M-2}{M-1-\mu} + \frac{\eps(\frac{1}{2}-\bar E_2)}{(1-\eps)(M-1-\mu) + \eps} \Bigg)^2.
    \end{align*}
    Then, note that 
    \begin{align}\label{eq: piimpl}
        \mu = \Bigg[1 + \exp\Big( \ut^{-1}\Delta(\mu)\Big) + \sum_{i=3}^K\exp\big(\ut^{-1}\Delta_{i,2}\big) \Bigg]^{-1},
    \end{align}
    where $\Delta(\mu)=v_1-v_2$, $\Delta_{i,2} = u(\bar x,m_i) - u(\bar x,m_2)$. As $\Delta_{i,2}<< 0 \le \Delta(\mu)$ for $i\ge 3$, the sum in \eqref{eq: piimpl} is negligible. However, $\mu$ is only implicitly pinned down. Note first that $\Delta(0)>0$ (for all $\varepsilon$ small enough), as in this case, $\bar x$ is part of $P'$, and the resulting posteriors are such that $m_2$ is preferred over $m_1$. Thus, at $\mu=0$, the RHS of \eqref{eq: piimpl} is positive. Next, note that the sender is only just indifferent between $P',P$ at state $\bar x$ if  $\bar x$ sends $P$s message with probability one. It is then not surprising, that
    \begin{align*}
        \frac{M}{M-1 + \mu}&>\frac{M-2}{M-1-\mu} \\
          M(M-1) - M\mu &> (M-2)(M-1) + (M-2)\mu\\
          1&>\mu,
    \end{align*}
    which implies that for all $c>0$ there exists $\bar \eps$ and $\delta_c$ s.t. for all $\varepsilon\le \bar \eps$, $\Delta(\mu)\ge c$ if and only if $\mu< 1-\delta_c$. $\delta_c$ vanishes as $c$ vanishes. In other words, as long as $\mu<1$, $\bar x$ would prefer to send $m_2$. This holds uniformly in $\ut$. As $\Delta>0$ implies that $RHS<\frac{1}{2}$, all fixed points of \eqref{eq: piimpl} lie strictly below $\frac{1}{2}$. It is then intuitive that any such fixed point must decrease towards zero as $\ut$ decreases, which we now verify. To ease notation, write $\partial_{\mu}, \partial^2_{\mu}$ as first and second derivatives w.r.t. $\mu$. 
    \begin{align*}
        \partial_{\mu}\bar E_1 &= \frac{M-1 }{2(K-1)}\frac{M}{(M-1+\mu)^2}\\
        \partial_{\mu}\bar E_2 &= \frac{M-1 }{2(K-1)}\frac{M-2}{(M-1-\mu)^2}
    \end{align*}
    \begin{align*}
        \partial_{\mu}\big(\bar x - y(m_1)\big) &= \frac{\eps(1-\eps)(\frac{1}{2}-\bar E_1)}{((1-\eps)(M-1+\mu) + \eps)^2}-\frac{M-1 }{2(K-1)}\frac{M}{(M-1+\mu)^2} \\
        &+\frac{\eps(\partial_{\mu}\bar E_1)}{(1-\eps)(M-1+\mu) + \eps}\\
        &= \frac{\eps(1-\eps)(\frac{1}{2}-\bar E_1)}{((1-\eps)(M-1+\mu) + \eps)^2}-\partial_{\mu}\bar E_1\frac{(1-\eps)(M-1+\mu)}{(1-\eps)(M-1+\mu) + \eps}\\
        &=\frac{(1-\eps)}{(1-\eps)(M-1+\mu) + \eps}\Bigg[\frac{\eps(\frac{1}{2}-\bar E_1)}{(1-\eps)(M-1+\mu) + \eps}  - \partial_{\mu}\bar E_1(M-1+\mu) \Bigg]\\
        &= \frac{(1-\eps)}{(1-\eps)(M-1+\mu) + \eps}\big( y(m_1)-\bar x\big),
    \end{align*}
    which is negative. Similarly, 
    \begin{align*}
        \partial_{\mu}\big(y(m_2)-\bar x\big) &= \partial_{\mu}\bar E_2\frac{(1-\eps)(M-1-\mu)}{(1-\eps)(M-1-\mu) + \eps} + \frac{\eps(1-\eps)(\frac{1}{2}-\bar E_2)}{((1-\eps)(M-1-\mu) + \eps)^2} \\
        &=\frac{(1-\eps)}{(1-\eps)(M-1-\mu) + \eps}\big( y(m_2)- \bar x\big),
    \end{align*}
    which is always positive, confirming the previous claims about $v_1,v_2$. It follows that
    \begin{align*}
        \partial_{\mu}\Delta(\mu) &= \partial_{\mu}v_1 - \partial_{\mu}v_2\\
        &= - 2(1-\eps)\Big( \frac{v_1}{(1-\eps)(M-1+\mu) + \eps} + \frac{v_2}{(1-\eps)(M-1-\mu) + \eps} \Big)  <0.
    \end{align*}
 Let $G(\mu)$ be the (approximate) RHS, discarding terms involving $\Delta_{i,2}$. Then, 
    \begin{align*}
        G'(\mu) = -G(\mu)^2\exp\left(\ut^{-1}\Delta \right)\ut^{-1}\partial_{\mu}\Delta = -G(\mu)\left(1-G(\mu) \right)\ut^{-1}\partial_{\mu}\Delta>0,
    \end{align*}
    as $\partial_{\mu}\Delta<0$. Next, note
    \begin{align*}
        \partial^2_{\mu}\big(\bar x - y(m_1)\big) &= \left(\frac{(1-\eps)}{(1-\eps)(M-1+\mu) + \eps} \right)^2\big(\bar x - y(m_1)\big) \\
        &- \frac{(1-\eps)}{(1-\eps)(M-1+\mu) + \eps}\partial_{\mu}\big(\bar x - y(m_1)\big)\\
        &= 2\left(\frac{(1-\eps)}{(1-\eps)(M-1+\mu) + \eps} \right)^2\big(\bar x - y(m_1)\big),
    \end{align*}
    and similarly,
     \begin{align*}
        \partial^2_{\mu}\big(y(m_2)-\bar x\big) &= 2\left(\frac{(1-\eps)}{(1-\eps)(M-1-\mu) + \eps} \right)^2\big(y(m_2)-\bar x\big).
    \end{align*}
    Putting things together, we can derive 
    \begin{align*}
        \partial^2_{\mu}v_1 &= \partial_{\mu}\left(2\big(\bar x - y(m_1)\big)\partial_{\mu}\big(\bar x - y(m_1)\big)
        \right)\\
        &= 6\left(\frac{(1-\eps)}{(1-\eps)(M-1+\mu) + \eps} \right)^2v_1,
    \end{align*}
    and similarly,
    \begin{align*}
        \partial^2_{\mu}v_2 &= 6\left(\frac{(1-\eps)}{(1-\eps)(M-1-\mu) + \eps} \right)^2v_2.
    \end{align*}
    Then,
    \begin{align*}
        \partial^2_{\mu}\Delta &= 6\left(\frac{(1-\eps)}{(1-\eps)(M-1+\mu) + \eps} \right)^2v_1- 6\left(\frac{(1-\eps)}{(1-\eps)(M-1-\mu) + \eps} \right)^2v_2\\
        &= 6(1-\eps)^2\left(\frac{v_1}{\left( (1-\eps)(M-1+\mu) + \eps\right)^2} -\frac{v_2}{\left( (1-\eps)(M-1-\mu) + \eps\right)^2}  \right).
    \end{align*}
    At $\varepsilon=0$, the term in brackets simplifies to 
    \begin{align*}
        \left(\frac{M-1 }{2(K-1)}\frac{M}{(M-1+\mu)^2} \right)^2 - \left(\frac{M-1 }{2(K-1)}\frac{M-2}{(M-1-\mu)^2} \right)^2.
    \end{align*}
    Note that 
    \begin{align*}
        G''(\mu) &= -\left(G'(\mu)-2G(\mu)G'(\mu) \right)\ut^{-1}\partial_{\mu}\Delta- G(\mu)\left(1-G(\mu) \right)\ut^{-1}\partial^2_{\mu}\Delta\\
        &= G'(\mu)\left(\left( 2G(\mu)-1\right)\ut^{-1}\partial_{\mu}\Delta + \frac{\partial^2_{\mu}\Delta}{\partial_{\mu}\Delta}\right)
        \end{align*}
        Thus, for $\ut$ small enough, the first term in the brackets must dominate, and $G''(\mu)>0$ for all $\mu<\frac{1}{2}$. It follows that, for all $\eps,\ut$ small enough, there is a unique fixed point $\mu^*$ to \eqref{eq: piimpl}, satisfying $\mu^*<\frac{1}{2}$. Note that the relationship between $\eps,\ut$ is unimportant here, it is enough that both variables be small enough.

We know that the fixed point $\mu^*$ must be bounded away from $\frac{1}{2}$ as $\Delta(\mu^*)\ge c>0$. Thus, as $\mu^*\approx \left(1+\exp(\ut^{-1}\Delta(\mu^*))\right)^{-1}$, we must have that $\frac{\mu^*}{\ut}\rightarrow0$ as $\ut\rightarrow0$. Any pool configuration $P,P'$ with $|P|=|P'|+2$ will be broken in favor of a configuration $|P|=|P'|$ as $\ut$ shrinks, completing the proof. 
\end{proof}

\subsection{Proof of Theorem \ref{thm:infobound}}
Suppose $K$ odd, so that there is a middle state $\bar x = (K+1)/2$. Theorem \ref{thm: attrchar} characterises MSFR and SAPS as constraints on policies that may be approached in the long run by the algorithmic sender. These constraints imply that any policy satisfying them must achieve point (1) for some level $\underline{U}_K\ge0$. We will now show points 2-3.  Without loss, consider policies symmetric around the middle state. We count pools of these policies from middle outwards, so that $N_0 = 1= |P_0|$, where $P_0=\{\bar x\}$ is the middle `pool'. Every policy will have some $n$ pools to the right, and to the left of $P_0$. By MSFR and SAPS, these policies must satisfy $N_i \le N_{i+1} +1$, and $N_i \le N_{i-1} +1$. For any $n$, let $\{\bar N_i\}_{i=1}^n$ be the strictly increasing sequence of pools $\bar N_i = \bar N_{i-1}+1$. Whether this policy is feasible depends on $n,K$. We have:
\begin{align*}
    \sum_{i=1}^n\bar N_i = \sum_{i=1}^ni+1 = n + \frac{n(n+1)}{2} = \frac{n(n+3)}{2}.
\end{align*}
This is feasible if $\frac{K-1}{2} = \frac{n(n+3)}{2}$, where $\frac{K-1}{2}$ is the number of states available to the right, left of the middle state. We get that if
\begin{align*}
    \hat n_K = -\frac{3}{2}+\sqrt{\frac{5}{4} + K}
\end{align*}
is an integer, the strictly increasing pool sequence is feasible. We then pin down the largest feasible pool size $\bar N_n = \hat n_K+1 $. By construction and the constraints on pool size increase, we can conclude that $M(K) = \lfloor \hat n_K +1 \rfloor$ is the largest feasible pool size for any $K$. Also, $n_K = \lfloor \hat n_K \rfloor$ is the number of strictly increasing pools $\bar N_i$ needed to achieve $M(K)= \bar N_{n_K}$.

We can write the Sender's payoffs as $U(\mu) = -\mathbb{E}\left(\left(x-\mathbb{E}(x\mid m,\mu)\right)^2\mid \mu\right)$, the residual variance over the state given policy $\mu$. To simplify this based on the above construction of policies, note that given any connected pool $P$ of length $N$, the conditional variance is that of a discrete uniform over the pool, i.e. $V(x\mid P) = \frac{1}{12}\frac{N^2-1}{(K-1)^2}$. The residual variance for policies consisting of connected pools is then the expected value of these conditional variances, assuming there are $n$ pools to either side of the middle state: 
\begin{align*}
    -U(\mu) &= 2\sum_{i=1}^n\frac{N_i}{K}\frac{1}{12}\frac{N_i^2-1}{(K-1)^2} = \frac{1}{6K(K-1)^2}\sum_{i=1}^n\left(N_i^3-N_i\right)\\
    &= \frac{1}{6K(K-1)^2}\sum_{i=1}^nN_i^3 - \frac{1}{12K(K-1)},
\end{align*}
as the sum of $N_i$ equals $\frac{K-1}{2}$. 

 Firstly, note that only sequences $\{N_i\}$ that are weakly monotone increasing can be potential minimizers. For any partition featuring some decreasing section, i.e. there is $j$ such that $N_j > N_{j+1}$, one can replace it with an alternative, feasible partition $\tilde N_i$ in the following way: Let $j>1$ be the first index at which $N_j > N_{j+1}$. Then $\tilde N_i = N_i$ for all $i< j$. Let $j'>j$ be the first $j'$ such that $N_{j'} = N_j$, if it exists. If $j'$ does not exist, let $\tilde N_i = N_j$ for all $i\ge j$. If $j'$ exists, let $\tilde N_i = N_j$ for $j\le i \le j'$. Let $\tilde N_i$ match $N_i$ whenever $N_i$ continues to weakly increase, then repeat the previous if there is another decrease.  Continue until $i=n$. Such $\{\tilde N_i\}$ must necessarily have a larger average $\frac{1}{n}\sum_i \tilde N_i$ than the original $N_i$. Letting $\mu$ implement the original sequence, and $\tilde \mu$ the newly constructed alternative, we note that $-U\left(\mu\right)$ can only shrink when replacing $N_i$ with weakly larger $\tilde N_i$, and so we must have $-U\left(\mu\right)\le -U\left(\{\tilde \mu\}\right)$.

 Consider the case of $\hat n_K = n_K$. The strictly increasing partition $\{\bar N_i\}$ is feasible. Any other feasible partition must have $n > n_K$, and feature some $i$ s.t. $N_i = N_{i+1}>1$. Let $i\ge1$ be smallest such index. First, if $N_i=1$, merge $P_i$ and $P_{i+1}$, leaving every other pool untouched. This will satisfy all equilibrium constraints, and by convexity of $x^3$ must increase the residual variance. Next, suppose $N_i>1$. Move one state in $P_i$ to $P_{i+1}$, and for all $j < i$, move one state from $P_j$ to $P_{j+1}$ as long as $N_j>1$, thereby creating a new partition $\{\tilde N_i\}$ which matches at all other indices. This again is a feasible partition as per our equilibrium constraints, albeit not necessarily weakly increasing. By convexity of $x^3$, $-U\left(\mu\right)\le -U\left(\{\tilde \mu\}\right)$ must hold. Thus, any such $\{N_i\}$ is improvable, while $\bar N_i$ is not. The strictly increasing partition minimizes $U()$ when $\hat n_K=n_K$.
 
For general $\hat n_K$, finding a minimizer of $U()$ given SAPS and MSFR is a combinatorial problem. One can bound minimal $U()$ by focusing on $\bar N_i$ however. Let $\{\bar N_i\}^-$ be the feasible strictly increasing partition adapted to any $\hat n_K$, so that if $\hat n_K$ is not an integer, $\{\bar N_i\}^-$ is constructed by starting at full revelation for the first $1\le i$, up until $j>1$, at which point strict increases are made until $\bar N_{n_K} =M(K)$ is reached. Also, let $\{\bar N_i\}^+$ be the infeasible strictly increasing sequence with $\bar N_1^+=2,\ldots \bar N_{n_K+1}=M(K)+1$.  Letting $U^+$ be the payoff for a policy based on $\bar N_i^+$, we get
\begin{align*}
    -U^+ &= 2\sum_{i=1}^{n_K+1}\frac{\bar N^+_i}{K}\frac{1}{12}\frac{(\bar N^+_i)^2-1}{(K-1)^2}\\
    &= \frac{1}{6K(K-1)^2}\sum_{i=1}^{n_K+1}(\bar N^+_i)^3 - \frac{1}{12K(K-1)}\\
    &=\frac{1}{6K(K-1)^2}\left(\left(\frac{(n_K+2)(n_K+3)}{2}\right)^2-1\right)- \frac{1}{12K(K-1)}.
\end{align*}
Now recall that by construction, $(n_K+1)(n_K+2)=n_K^2+3n_K+2=K+1$. So, that $(n_K+2)(n_K+3) = K+2n_K+5$. 
\begin{align*}
    -U^+ &= \frac{1}{6K(K-1)^2}\left( \frac{K+2n_K +3}{2}\right)\left( \frac{(K+2n_K +7}{2}\right)- \frac{1}{12K(K-1)},
\end{align*}
where we have that the highest order term in $K$ multiplying $1/(K(K-1)^2))$ is proportional to $K^2$, since $n_K = O(\sqrt{K})$. . Hence we get that $-U^+ = O(1/K)$. 
Similarly, let $U^-$ be the payoff for a policy based on $\bar N_i^-$, and we get
\begin{align*}
    -U^- &= 2\sum_{i=1}^{n_K}\frac{\bar N^-_i}{K}\frac{1}{12}\frac{(\bar N^-_i)^2-1}{(K-1)^2}\\ 
    &=-U^+ - \frac{((n_K+1)^3-n_K-1)}{6K(K-1)^2}.
\end{align*}
By construction $U^+\le U^*\le U^-$ for all $K$, where $U^*$ is the true minimal utility given SAPS and MSFR. We have that $\frac{U^+}{U^-}\to 1$ as $K\to \infty$. We have that $U^+ \to0$ as $K\to \infty$ at a rate of $1/K$,  as required. This comes as no surprise when considering that the largest pool size in our worst-policy construction is $n_K+1$, hence only grows  with $\sqrt{K}$. \hfill $\square$

\subsection{Proof of Theorem \ref{thm: Posbiasnolearn}}
First, note that $b>\underline{b}_K$ implies that for all states $x<x_K=1$,  when $\eps>0$ is small enough, pooling with a state higher than $x$ is preferred to getting a full-revelation payoff. Hence, any $\mu \in \mathbf{NE}_{CS}$ will be a non-trivial partition, where each state must be pooled with some other states. Then, if the equilibrium policy involves mixing over multiple messages for any pool, it must be unstable by Lemma \ref{lem: nomixpools}. If instead the policy is pure, there will be a nonempty set $\mathbf{M}_0$ of unused messages. By Lemma \ref{lem: unusedM}, for $\ut$ small enough, the posterior attached to such messages will be close to $1/2$. For any bias $b>0$, there will be a state $x\approx 1/2 - b$ for which some unused message then will provide a profitable deviation, as the associated payoff would be closer to zero than any payoff this state would receive under the pool assigned by the equilibrium policy. Hence, a pure equilibrium policy cannot be learnt by the learner either. In the knife-edge case where such state would be indifferent between sending unused messages and their equilibrium pool message, we'd again reach a contradiction: the Q-learner would have to learn to mix equally over all such previously unused messages and the original pooling messages, leading to instability as before for the case pools under message mixtures. \hfill $\square$

\section{Generalization to learning receiver}\label{app: leanringR}
We'd like to generalize our model to one where the receiver may adjust their behavior over time, not being `correct' at every finite $t$. We show here that Proposition \ref{prop: stochapproxQ} extends to such a setting, and hence all other results in the paper do too. A straightforward generalization of our result is one where the receiver learns, but does so at a faster speed than the sender. Let the sender's $Q$-updating scheme be as before, fixing a stepsize-schedule $\alpha_t$. Now suppose that in addition, the receiver updates their posterior $y_t(m)$, for all $m\in \mathbf{M}_K$, the following way: 
\begin{align}\label{eq: R_learn}
    y_{t+1}(m) = (1-\beta_t)y_t(m) + \beta_t\left(y(m,Q_t)+R_{t+1} \right),
\end{align}
where $\beta_t\rightarrow 0$ under the usual Robins-Monro conditions as satisfied by $\alpha_t$, and in addition $\frac{\alpha_t}{\beta_t}\rightarrow0$. This additional assumption implies that asymptotically, the sender's updates to the Q-table are slower than the receiver's updates to their posterior. Letting $\mathcal{F}_t = \sigma\left(\left\{y_k,Q_k,R_k\right\}_{k=0}^t\right)$ be the sigma algebra generated by the algorithm's trajectory up until $t$, we assume $R_{t+1}$ a Martingale-difference error according to the receiver's filtration $\mathcal{F}_t$. We assume that $R_{t+1}$ has bounded second moments given $\mathcal{F}_t$. $y(m,Q_t)$ is the correct posterior given matrix $Q_t$. Thus, \eqref{eq: R_learn} is a reduced-form equation modeling any updating scheme that in the $t$-limit will be arbitrarily close to a correct update based on the posterior given $Q_t$. Formally, consider the coupled system implied by this new setting:
\begin{align*}
    Q_{t+1}(x,m) &= Q_t(x,m) + \alpha_t\left(u^H(x,m,y_t(m))-Q_t(x,m) + M_{t+1}\right)\\
    y_{t+1}(m) &= y_t(m) + \beta_t\left(y(m,Q_t)-y_t(m)+R_{t+1} \right),
\end{align*}
which we can re-write as 
\begin{align*}
    Q_{t+1}(x,m) &= Q_t(x,m) + \beta_t\frac{\alpha_t}{\beta_t}\left(u^H(x,m,y_t(m))-Q_t(x,m) + M_{t+1}\right)\\
    y_{t+1}(m) &= y_t(m) + \beta_t\left(y(m,Q_t)-y_t(m)+R_{t+1} \right),
\end{align*} and note that as $\frac{\alpha_t}{\beta_t}\to 0$, $Q_t$ is approximately constant according to the timescale of $y_t$, for $t$ large enough. Note that for any fixed matrix $\bar Q$, trajectories of \eqref{eq: R_learn} will approximate solutions to the ODE (by an argument analogous to the on in the proof of Proposition \ref{prop: stochapproxQ})
\begin{align*}
    \dot{\tilde y}(m) = y(m,\bar Q) - \tilde y(m). 
\end{align*}
For any fixed $\bar Q$, the above has a globally attracting fixed point at $y(m,\bar Q)$, the correct posterior given $\bar Q$. Hence, we can apply the analysis and results of \cite[Chapter 6]{borkar2008stochastic} to deduce that, just like in the proof of Proposition \ref{prop: stochapproxQ}, the limiting O.D.E approximated by the Sender's Q-learning algorithm is \eqref{eq: ODE}. The remaining points regarding attracting sets in the Proposition then readily extend to this case, by writing the sender's $Q$ iteration \eqref{eq: Qlearnprecise} with an additional error term equal $U^S(y(m,Q_t),x,b)-U^S(y_t(m),x,b)$, which vanishes according to the $\alpha_t$-timescale of the $Q$-iteration, again since the receiver's learning speed $\beta_t$ is an order of magnitude faster than the sender's.

\end{document}